\documentclass[transmag]{IEEEtran}

\usepackage{cite}
\usepackage[T1]{fontenc}
\usepackage{graphicx}
\usepackage{amssymb}
\usepackage{amsmath}
\usepackage{subfigure}
\usepackage{booktabs} 
\usepackage{amsthm}
\usepackage{multirow}
\usepackage{microtype} 
\usepackage{tablefootnote}
\usepackage{balance}
\usepackage{xcolor}
\usepackage{float}
\usepackage{hyperref}

\usepackage{amssymb}
\usepackage{amsfonts}
\usepackage{mathrsfs}
\usepackage{xspace}
\usepackage{bm}
\usepackage{upgreek}

\newcommand{\safemath}[2]{\newcommand{#1}{\ensuremath{#2}\xspace}}

\safemath{\bma}{\mathbf{a}}
\safemath{\bmb}{\mathbf{b}}
\safemath{\bmc}{\mathbf{c}}
\safemath{\bmd}{\mathbf{d}}
\safemath{\bme}{\mathbf{e}}
\safemath{\bmf}{\mathbf{f}}
\safemath{\bmg}{\mathbf{g}}
\safemath{\bmh}{\mathbf{h}}
\safemath{\bmi}{\mathbf{i}}
\safemath{\bmj}{\mathbf{j}}
\safemath{\bmk}{\mathbf{k}}
\safemath{\bml}{\mathbf{l}}
\safemath{\bmm}{\mathbf{m}}
\safemath{\bmn}{\mathbf{n}}
\safemath{\bmo}{\mathbf{o}}
\safemath{\bmp}{\mathbf{p}}
\safemath{\bmq}{\mathbf{q}}
\safemath{\bmr}{\mathbf{r}}
\safemath{\bms}{\mathbf{s}}
\safemath{\bmt}{\mathbf{t}}
\safemath{\bmu}{\mathbf{u}}
\safemath{\bmv}{\mathbf{v}}
\safemath{\bmw}{\mathbf{w}}
\safemath{\bmx}{\mathbf{x}}
\safemath{\bmy}{\mathbf{y}}
\safemath{\bmz}{\mathbf{z}}
\safemath{\bmzero}{\mathbf{0}}
\safemath{\bmone}{\mathbf{1}}

\bmdefine{\biad}{a}
\bmdefine{\bibd}{b}
\bmdefine{\bicd}{c}
\bmdefine{\bidd}{d}
\bmdefine{\bied}{e}
\bmdefine{\bifd}{f}
\bmdefine{\bigd}{g}
\bmdefine{\bihd}{h}
\bmdefine{\biid}{i}
\bmdefine{\bijd}{j}
\bmdefine{\bikd}{k}
\bmdefine{\bild}{l}
\bmdefine{\bimd}{m}
\bmdefine{\bind}{n}
\bmdefine{\biod}{o}
\bmdefine{\bipd}{p}
\bmdefine{\biqd}{q}
\bmdefine{\bird}{r}
\bmdefine{\bisd}{s}
\bmdefine{\bitd}{t}
\bmdefine{\biud}{u}
\bmdefine{\bivd}{v}
\bmdefine{\biwd}{w}
\bmdefine{\bixd}{x}
\bmdefine{\biyd}{y}
\bmdefine{\bizd}{z}

\bmdefine{\bixid}{\xi}
\bmdefine{\bilambdad}{\lambda}
\bmdefine{\bimud}{\mu}
\bmdefine{\bithetad}{\theta}
\bmdefine{\biphid}{\phi}
\bmdefine{\bideltad}{\delta}

\safemath{\bmia}{\biad}
\safemath{\bmib}{\bibd}
\safemath{\bmic}{\bicd}
\safemath{\bmid}{\bidd}
\safemath{\bmie}{\bied}
\safemath{\bmif}{\bifd}
\safemath{\bmig}{\bigd}
\safemath{\bmih}{\bihd}
\safemath{\bmii}{\biid}
\safemath{\bmij}{\bijd}
\safemath{\bmik}{\bikd}
\safemath{\bmil}{\bild}
\safemath{\bmim}{\bimd}
\safemath{\bmin}{\bind}
\safemath{\bmio}{\biod}
\safemath{\bmip}{\bipd}
\safemath{\bmiq}{\biqd}
\safemath{\bmir}{\bird}
\safemath{\bmis}{\bisd}
\safemath{\bmit}{\bitd}
\safemath{\bmiu}{\biud}
\safemath{\bmiv}{\bivd}
\safemath{\bmiw}{\biwd}
\safemath{\bmix}{\bixd}
\safemath{\bmiy}{\biyd}
\safemath{\bmiz}{\bizd}

\safemath{\bmxi}{\bixid}
\safemath{\bmlambda}{\bilambdad}
\safemath{\bmmu}{\bimud}
\safemath{\bmtheta}{\bithetad}
\safemath{\bmphi}{\biphid}
\safemath{\bmdelta}{\bideltad}

\safemath{\bA}{\mathbf{A}}
\safemath{\bB}{\mathbf{B}}
\safemath{\bC}{\mathbf{C}}
\safemath{\bD}{\mathbf{D}}
\safemath{\bE}{\mathbf{E}}
\safemath{\bF}{\mathbf{F}}
\safemath{\bG}{\mathbf{G}}
\safemath{\bH}{\mathbf{H}}
\safemath{\bI}{\mathbf{I}}
\safemath{\bJ}{\mathbf{J}}
\safemath{\bK}{\mathbf{K}}
\safemath{\bL}{\mathbf{L}}
\safemath{\bM}{\mathbf{M}}
\safemath{\bN}{\mathbf{N}}
\safemath{\bO}{\mathbf{O}}
\safemath{\bP}{\mathbf{P}}
\safemath{\bQ}{\mathbf{Q}}
\safemath{\bR}{\mathbf{R}}
\safemath{\bS}{\mathbf{S}}
\safemath{\bT}{\mathbf{T}}
\safemath{\bU}{\mathbf{U}}
\safemath{\bV}{\mathbf{V}}
\safemath{\bW}{\mathbf{W}}
\safemath{\bX}{\mathbf{X}}
\safemath{\bY}{\mathbf{Y}}
\safemath{\bZ}{\mathbf{Z}}

\safemath{\bZero}{\mathbf{0}}
\safemath{\bOne}{\mathbf{1}}
\safemath{\bDelta}{\mathbf{\Delta}}
\safemath{\bLambda}{\mathbf{\UpLambda}}
\safemath{\bPhi}{\mathbf{\Upphi}}
\safemath{\bSigma}{\mathbf{\Upsigma}}
\safemath{\bOmega}{\mathbf{\Upomega}}
\safemath{\bTheta}{\mathbf{\Uptheta}}

\bmdefine{\biAd}{A}
\bmdefine{\biBd}{B}
\bmdefine{\biCd}{C}
\bmdefine{\biDd}{D}
\bmdefine{\biEd}{E}
\bmdefine{\biFd}{F}
\bmdefine{\biGd}{G}
\bmdefine{\biHd}{H}
\bmdefine{\biId}{I}
\bmdefine{\biJd}{J}
\bmdefine{\biKd}{K}
\bmdefine{\biLd}{L}
\bmdefine{\biMd}{M}
\bmdefine{\biOd}{N}
\bmdefine{\biPd}{O}
\bmdefine{\biQd}{P}
\bmdefine{\biRd}{R}
\bmdefine{\biSd}{S}
\bmdefine{\biTd}{T}
\bmdefine{\biUd}{U}
\bmdefine{\biVd}{V}
\bmdefine{\biWd}{W}
\bmdefine{\biXd}{X}
\bmdefine{\biYd}{Y}
\bmdefine{\biZd}{Z}

\bmdefine{\biDelta}{\Delta}
\bmdefine{\biLambda}{\Lambda}
\bmdefine{\biPhi}{\Phi}
\bmdefine{\biSigma}{\Sigma}
\bmdefine{\biOmega}{\Omega}
\bmdefine{\biTheta}{\Theta}

\safemath{\bimA}{\biAd}
\safemath{\bimB}{\biBd}
\safemath{\bimC}{\biCd}
\safemath{\bimD}{\biDd}
\safemath{\bimE}{\biEd}
\safemath{\bimF}{\biFd}
\safemath{\bimG}{\biGd}
\safemath{\bimH}{\biHd}
\safemath{\bimI}{\biId}
\safemath{\bimJ}{\biJd}
\safemath{\bimK}{\biKd}
\safemath{\bimL}{\biLd}
\safemath{\bimM}{\biMd}
\safemath{\bimN}{\biNd}
\safemath{\bimO}{\biOd}
\safemath{\bimP}{\biPd}
\safemath{\bimQ}{\biQd}
\safemath{\bimR}{\biRd}
\safemath{\bimS}{\biSd}
\safemath{\bimT}{\biTd}
\safemath{\bimU}{\biUd}
\safemath{\bimV}{\biVd}
\safemath{\bimW}{\biWd}
\safemath{\bimX}{\biXd}
\safemath{\bimY}{\biYd}
\safemath{\bimZ}{\biZd}

\safemath{\bimDelta}{\biDelta}
\safemath{\bimLambda}{\biLambda}
\safemath{\bimPhi}{\biPhi}
\safemath{\bimSigma}{\biSigma}
\safemath{\bimOmega}{\biOmega}
\safemath{\bimTheta}{\biTheta}

\safemath{\setA}{\mathcal{A}}
\safemath{\setB}{\mathcal{B}}
\safemath{\setC}{\mathcal{C}}
\safemath{\setD}{\mathcal{D}}
\safemath{\setE}{\mathcal{E}}
\safemath{\setF}{\mathcal{F}}
\safemath{\setG}{\mathcal{G}}
\safemath{\setH}{\mathcal{H}}
\safemath{\setI}{\mathcal{I}}
\safemath{\setJ}{\mathcal{J}}
\safemath{\setK}{\mathcal{K}}
\safemath{\setL}{\mathcal{L}}
\safemath{\setM}{\mathcal{M}}
\safemath{\setN}{\mathcal{N}}
\safemath{\setO}{\mathcal{O}}
\safemath{\setP}{\mathcal{P}}
\safemath{\setQ}{\mathcal{Q}}
\safemath{\setR}{\mathcal{R}}
\safemath{\setS}{\mathcal{S}}
\safemath{\setT}{\mathcal{T}}
\safemath{\setU}{\mathcal{U}}
\safemath{\setV}{\mathcal{V}}
\safemath{\setW}{\mathcal{W}}
\safemath{\setX}{\mathcal{X}}
\safemath{\setY}{\mathcal{Y}}
\safemath{\setZ}{\mathcal{Z}}
\safemath{\emptySet}{\varnothing}

\safemath{\colA}{\mathscr{A}}
\safemath{\colB}{\mathscr{B}}
\safemath{\colC}{\mathscr{C}}
\safemath{\colD}{\mathscr{D}}
\safemath{\colE}{\mathscr{E}}
\safemath{\colF}{\mathscr{F}}
\safemath{\colG}{\mathscr{G}}
\safemath{\colH}{\mathscr{H}}
\safemath{\colI}{\mathscr{I}}
\safemath{\colJ}{\mathscr{J}}
\safemath{\colK}{\mathscr{K}}
\safemath{\colL}{\mathscr{L}}
\safemath{\colM}{\mathscr{M}}
\safemath{\colN}{\mathscr{N}}
\safemath{\colO}{\mathscr{O}}
\safemath{\colP}{\mathscr{P}}
\safemath{\colQ}{\mathscr{Q}}
\safemath{\colR}{\mathscr{R}}
\safemath{\colS}{\mathscr{S}}
\safemath{\colT}{\mathscr{T}}
\safemath{\colU}{\mathscr{U}}
\safemath{\colV}{\mathscr{V}}
\safemath{\colW}{\mathscr{W}}
\safemath{\colX}{\mathscr{X}}
\safemath{\colY}{\mathscr{Y}}
\safemath{\colZ}{\mathscr{Z}}

\safemath{\opA}{\mathbb{A}}
\safemath{\opB}{\mathbb{B}}
\safemath{\opC}{\mathbb{C}}
\safemath{\opD}{\mathbb{D}}
\safemath{\opE}{\mathbb{E}}
\safemath{\opF}{\mathbb{F}}
\safemath{\opG}{\mathbb{G}}
\safemath{\opH}{\mathbb{H}}
\safemath{\opI}{\mathbb{I}}
\safemath{\opJ}{\mathbb{J}}
\safemath{\opK}{\mathbb{K}}
\safemath{\opL}{\mathbb{L}}
\safemath{\opM}{\mathbb{M}}
\safemath{\opN}{\mathbb{N}}
\safemath{\opO}{\mathbb{O}}
\safemath{\opP}{\mathbb{P}}
\safemath{\opQ}{\mathbb{Q}}
\safemath{\opR}{\mathbb{R}}
\safemath{\opS}{\mathbb{S}}
\safemath{\opT}{\mathbb{T}}
\safemath{\opU}{\mathbb{U}}
\safemath{\opV}{\mathbb{V}}
\safemath{\opW}{\mathbb{W}}
\safemath{\opX}{\mathbb{X}}
\safemath{\opY}{\mathbb{Y}}
\safemath{\opZ}{\mathbb{Z}}
\safemath{\opZero}{\mathbb{O}}
\safemath{\identityop}{\opI}

\safemath{\veca}{\bma}
\safemath{\vecb}{\bmb}
\safemath{\vecc}{\bmc}
\safemath{\vecd}{\bmd}
\safemath{\vece}{\bme}
\safemath{\vecf}{\bmf}
\safemath{\vecg}{\bmg}
\safemath{\vech}{\bmh}
\safemath{\veci}{\bmi}
\safemath{\vecj}{\bmj}
\safemath{\veck}{\bmk}
\safemath{\vecl}{\bml}
\safemath{\vecm}{\bmm}
\safemath{\vecn}{\bmn}
\safemath{\veco}{\bmo}
\safemath{\vecp}{\bmp}
\safemath{\vecq}{\bmq}
\safemath{\vecr}{\bmr}
\safemath{\vecs}{\bms}
\safemath{\vect}{\bmt}
\safemath{\vecu}{\bmu}
\safemath{\vecv}{\bmv}
\safemath{\vecw}{\bmw}
\safemath{\vecx}{\bmx}
\safemath{\vecy}{\bmy}
\safemath{\vecz}{\bmz}

\safemath{\veczero}{\bmzero}
\safemath{\vecone}{\bmone}
\safemath{\vecxi}{\bmxi}
\safemath{\veclambda}{\bmlambda}
\safemath{\vecmu}{\bmmu}
\safemath{\vectheta}{\bmtheta}
\safemath{\vecphi}{\bmphi}
\safemath{\vecdelta}{\bmdelta}

\safemath{\matA}{\bA}
\safemath{\matB}{\bB}
\safemath{\matC}{\bC}
\safemath{\matD}{\bD}
\safemath{\matE}{\bE}
\safemath{\matF}{\bF}
\safemath{\matG}{\bG}
\safemath{\matH}{\bH}
\safemath{\matI}{\bI}
\safemath{\matJ}{\bJ}
\safemath{\matK}{\bK}
\safemath{\matL}{\bL}
\safemath{\matM}{\bM}
\safemath{\matN}{\bN}
\safemath{\matO}{\bO}
\safemath{\matP}{\bP}
\safemath{\matQ}{\bQ}
\safemath{\matR}{\bR}
\safemath{\matS}{\bS}
\safemath{\matT}{\bT}
\safemath{\matU}{\bU}
\safemath{\matV}{\bV}
\safemath{\matW}{\bW}
\safemath{\matX}{\bX}
\safemath{\matY}{\bY}
\safemath{\matZ}{\bZ}
\safemath{\matzero}{\bmzero}

\safemath{\matDelta}{\bDelta}
\safemath{\matLambda}{\bLambda}
\safemath{\matPhi}{\bPhi}
\safemath{\matSigma}{\bSigma}
\safemath{\matOmega}{\bOmega}
\safemath{\matTheta}{\bTheta}

\safemath{\matidentity}{\matI}
\safemath{\matone}{\matO}

\safemath{\rnda}{A}
\safemath{\rndb}{B}
\safemath{\rndc}{C}
\safemath{\rndd}{D}
\safemath{\rnde}{E}
\safemath{\rndf}{F}
\safemath{\rndg}{G}
\safemath{\rndh}{H}
\safemath{\rndi}{I}
\safemath{\rndj}{J}
\safemath{\rndk}{K}
\safemath{\rndl}{L}
\safemath{\rndm}{M}
\safemath{\rndn}{N}
\safemath{\rndo}{O}
\safemath{\rndp}{P}
\safemath{\rndq}{Q}
\safemath{\rndr}{R}
\safemath{\rnds}{S}
\safemath{\rndt}{T}
\safemath{\rndu}{U}
\safemath{\rndv}{V}
\safemath{\rndw}{W}
\safemath{\rndx}{X}
\safemath{\rndy}{Y}
\safemath{\rndz}{Z}

\safemath{\rveca}{\bimA}
\safemath{\rvecb}{\bimB}
\safemath{\rvecc}{\bimC}
\safemath{\rvecd}{\bimD}
\safemath{\rvece}{\bimE}
\safemath{\rvecf}{\bimF}
\safemath{\rvecg}{\bimG}
\safemath{\rvech}{\bimH}
\safemath{\rveci}{\bimI}
\safemath{\rvecj}{\bimJ}
\safemath{\rveck}{\bimK}
\safemath{\rvecl}{\bimL}
\safemath{\rvecm}{\bimM}
\safemath{\rvecn}{\bimN}
\safemath{\rveco}{\bomO}
\safemath{\rvecp}{\bimP}
\safemath{\rvecq}{\bimQ}
\safemath{\rvecr}{\bimR}
\safemath{\rvecs}{\bimS}
\safemath{\rvect}{\bimT}
\safemath{\rvecu}{\bimU}
\safemath{\rvecv}{\bimV}
\safemath{\rvecw}{\bimW}
\safemath{\rvecx}{\bimX}
\safemath{\rvecy}{\bimY}
\safemath{\rvecz}{\bimZ}

\safemath{\rvecxi}{\bmxi}
\safemath{\rveclambda}{\bmlambda}
\safemath{\rvecmu}{\bmmu}
\safemath{\rvectheta}{\bmtheta}
\safemath{\rvecphi}{\bmphi}

\safemath{\rmatA}{\bimA}
\safemath{\rmatB}{\bimB}
\safemath{\rmatC}{\bimC}
\safemath{\rmatD}{\bimD}
\safemath{\rmatE}{\bimE}
\safemath{\rmatF}{\bimF}
\safemath{\rmatG}{\bimG}
\safemath{\rmatH}{\bimH}
\safemath{\rmatI}{\bimI}
\safemath{\rmatJ}{\bimJ}
\safemath{\rmatK}{\bimK}
\safemath{\rmatL}{\bimL}
\safemath{\rmatM}{\bimM}
\safemath{\rmatN}{\bimN}
\safemath{\rmatO}{\bimO}
\safemath{\rmatP}{\bimP}
\safemath{\rmatQ}{\bimQ}
\safemath{\rmatR}{\bimR}
\safemath{\rmatS}{\bimS}
\safemath{\rmatT}{\bimT}
\safemath{\rmatU}{\bimU}
\safemath{\rmatV}{\bimV}
\safemath{\rmatW}{\bimW}
\safemath{\rmatX}{\bimX}
\safemath{\rmatY}{\bimY}
\safemath{\rmatZ}{\bimZ}

\safemath{\rmatDelta}{\bimDelta}
\safemath{\rmatLambda}{\bimLambda}
\safemath{\rmatPhi}{\bimPhi}
\safemath{\rmatSigma}{\bimSigma}
\safemath{\rmatOmega}{\bimOmega}
\safemath{\rmatTheta}{\bimTheta}

\usepackage{amssymb}
\usepackage{amsfonts}
\usepackage{mathrsfs}
\usepackage{xspace}
\usepackage{bm}
\usepackage{fancyref}
\usepackage{textcomp}

\usepackage{multirow}
\usepackage{stmaryrd}

\newenvironment{textbmatrix}{	\setlength{\arraycolsep}{2.5pt}%
								\big[\begin{matrix}}{\end{matrix}\big]%
								\raisebox{0.08ex}{\vphantom{M}}}

\def\be{\begin{equation}}
\def\ee{\end{equation}}
\def\een{\nonumber \end{equation}}
\def\mat{\begin{bmatrix}}
\def\emat{\end{bmatrix}}
\def\btm{\begin{textbmatrix}}
\def\etm{\end{textbmatrix}}

\def\ba#1\ea{\begin{align}#1\end{align}}
\def\bas#1\eas{\begin{align*}#1\end{align*}}
\def\bs#1\es{\begin{split}#1\end{split}}
\def\bg#1\eg{\begin{gather}#1\end{gather}}
\def\bml#1\eml{\begin{multline}#1\end{multline}}
\def\bi#1\ei{\begin{itemize}#1\end{itemize}}

\newcommand{\lefto}{\mathopen{}\left}

\DeclareMathOperator{\tr}{tr}				
\DeclareMathOperator{\Exop}{\opE}			

\newcommand{\Ex}[2]{\ensuremath{\Exop_{#1}\lefto[#2\right]}} 	

\newcommand{\tp}[1]{\ensuremath{#1^{T}}} 		
\newcommand{\herm}[1]{\ensuremath{#1^{H}}} 	
\newcommand{\inv}[1]{\ensuremath{#1^{-1}}} 	

\safemath{\dirac}{\delta}					
\safemath{\krond}{\dirac}					

\safemath{\upto}{\uparrow}
\safemath{\downto}{\downarrow}
\safemath{\iu}{j}							
\safemath{\ev}{\lambda}						
\safemath{\hilseqspace}{l^{2}}				
\newcommand{\banachfunspace}[1]{\setL^{#1}}	
\safemath{\hilfunspace}{\banachfunspace{2}}	
\newcommand{\floor}[1]{\lfloor #1 \rfloor}

\safemath{\SNR}{\textit{SNR}} 				
\safemath{\PAR}{\textit{PAR}} 				
\safemath{\No}{N_0}							
\safemath{\Es}{E_s}							
\safemath{\Eb}{E_b}							
\safemath{\EbNo}{\frac{\Eb}{\No}}
\safemath{\EsNo}{\frac{\Es}{\No}}

\DeclareMathOperator{\CHop}{\ensuremath{\opH}} 
\safemath{\tvir}{\rndh_{\CHop}}				
\safemath{\tvtf}{\rndl_{\CHop}}				
\safemath{\spf}{\rnds_{\CHop}}				
\safemath{\bff}{H_{\CHop}}					

\safemath{\ircf}{r_{h}}						
\safemath{\tftvcf}{r_{s}}					
\safemath{\tfcf}{r_{l}}						
\safemath{\bfcf}{r_{H}}						

\safemath{\tcorr}{c_h}						
\safemath{\scf}{c_{s}}						
\safemath{\tfcorr}{c_{l}}					
\safemath{\fcorr}{c_{H}}						

\safemath{\mi}{I}							
\safemath{\capacity}{C}						

\safemath{\normal}{\mathcal{N}}			
\safemath{\jpg}{\mathcal{CN}}			
\safemath{\mchain}{\leftrightarrow}		

\safemath{\dB}{\,\mathrm{dB}}
\safemath{\dBm}{\,\mathrm{dBm}}
\safemath{\Hz}{\,\mathrm{Hz}}
\safemath{\kHz}{\,\mathrm{kHz}}
\safemath{\MHz}{\,\mathrm{MHz}}
\safemath{\GHz}{\,\mathrm{GHz}}
\safemath{\s}{\,\mathrm{s}}
\safemath{\ms}{\,\mathrm{ms}}
\safemath{\mus}{\,\mathrm{\text{\textmu}s}}
\safemath{\ns}{\,\mathrm{ns}}
\safemath{\ps}{\,\mathrm{ps}}
\safemath{\meter}{\,\mathrm{m}}
\safemath{\mm}{\,\mathrm{mm}}
\safemath{\cm}{\,\mathrm{cm}}
\safemath{\m}{\,\mathrm{m}}
\safemath{\W}{\,\mathrm{W}}
\safemath{\mW}{\, \mathrm{mW}}
\safemath{\J}{\,\mathrm{J}}
\safemath{\K}{\,\mathrm{K}}
\safemath{\bit}{\,\mathrm{bit}}
\safemath{\nat}{\,\mathrm{nat}}

\safemath{\define}{\triangleq}			

\safemath{\equivalent}{\sim}
\safemath{\distas}{\sim}					
\safemath{\sdiff}{\Delta}				

\safemath{\reals}{\mathbb{R}}
\safemath{\positivereals}{\reals_{+}}
\safemath{\integers}{\mathbb{Z}}
\safemath{\posint}{\integers_{+}}
\safemath{\naturals}{\mathbb{N}}
\safemath{\posnaturals}{\naturals_{+}}
\safemath{\complexset}{\mathbb{C}}
\safemath{\rationals}{\mathbb{Q}}

\newcommand*{\fancyrefapplabelprefix}{app}		
\newcommand*{\fancyrefthmlabelprefix}{thm}		
\newcommand*{\fancyreflemlabelprefix}{lem}		
\newcommand*{\fancyrefcorlabelprefix}{cor}		
\newcommand*{\fancyrefdeflabelprefix}{def}		
\newcommand*{\fancyrefproplabelprefix}{prop}		
\newcommand*{\fancyrefexmpllabelprefix}{exmpl}
\newcommand*{\fancyrefalglabelprefix}{alg}		
\newcommand*{\fancyreftbllabelprefix}{tbl}		

\frefformat{vario}{\fancyrefseclabelprefix}{Section~#1}
\frefformat{vario}{\fancyrefthmlabelprefix}{Theorem~#1}
\frefformat{vario}{\fancyreftbllabelprefix}{Table~#1}
\frefformat{vario}{\fancyreflemlabelprefix}{Lemma~#1}
\frefformat{vario}{\fancyrefcorlabelprefix}{Corollary~#1}
\frefformat{vario}{\fancyrefdeflabelprefix}{Definition~#1}
\frefformat{vario}{\fancyreffiglabelprefix}{Figure~#1}
\frefformat{vario}{\fancyrefapplabelprefix}{Appendix~#1}
\frefformat{vario}{\fancyrefeqlabelprefix}{(#1)}
\frefformat{vario}{\fancyrefproplabelprefix}{Proposition~#1}
\frefformat{vario}{\fancyrefexmpllabelprefix}{Example~#1}
\frefformat{vario}{\fancyrefalglabelprefix}{Algorithm~#1}

 \newtheorem{thm}{Theorem}

 \newtheorem{lem}[thm]{Lemma}
 
\safemath{\dictab}{[\,\dicta\,\,\dictb\,]}

\safemath{\ysig}{\bmy}
\safemath{\ysighat}{\hat{\ysig}}
\safemath{\ysigdim}{M}
\safemath{\xsig}{\bmx}
\safemath{\xsigdim}{N}
\safemath{\nx}{n_x}
\safemath{\zsig}{\bmz}
\safemath{\zsigdim}{\ysigdim}
\safemath{\rsig}{\bmr}
\safemath{\Adict}{\bA}
\safemath{\Adicttilde}{\widetilde{\Adict}}
\safemath{\Adictdim}{\outputdim\times\xsigdim}
\safemath{\avec}{\bma}
\safemath{\avectilde}{\tilde{\avec}}
\safemath{\Bdict}{\bB}
\safemath{\Bdicttilde}{\widetilde{\Bdict}}
\safemath{\Cdict}{\bC}
\safemath{\cvec}{\bmc}
\safemath{\Ddict}{\bD}
\safemath{\Ddictdim}{\ysigdim\times\xsigdim}
\safemath{\dvec}{\bmd}
\safemath{\Ddicttilde}{\widetilde{\bD}}
\safemath{\Bonb}{\bB}
\safemath{\bvec}{\bmb}
\safemath{\Bonbdim}{\ysigdim\times\ysigdim}
\safemath{\noise}{\bmn}
\safemath{\noisedim}{\ysigim}
\safemath{\err}{\bme}
\safemath{\errdim}{\ysigdim}
\safemath{\errset}{\setE}
\safemath{\nerr}{n_e}
\safemath{\delop}{\bP_\errset}
\safemath{\delopc}{\bP_{{\errset}^c}}

\safemath{\cplxi}{\imath}
\safemath{\cplxj}{\jmath}

\safemath{\dict}{\matD}
\safemath{\inputdim}{N}		
\safemath{\outputdim}{M}		
\safemath{\sparsity}{S}	
\safemath{\inputdimA}{{N_a}}	
\safemath{\inputdimB}{{N_b}}	
\safemath{\elemA}{{n_a}}	
\safemath{\elemB}{{n_b}}	
\safemath{\resA}{\matR_a}	
\safemath{\resB}{\matR_b}	
\safemath{\subD}{\matS} 
\safemath{\subA}{\matS_a} 
\safemath{\subB}{\matS_b} 
\safemath{\dicta}{\matA} 	
\safemath{\dictb}{\matB} 	
\safemath{\hollowS}{H}
\safemath{\hollowA}{H_a}
\safemath{\hollowB}{H_b}
\safemath{\cross}{Z}
\safemath{\coh}{\mu_d}			
\safemath{\coha}{\mu_a}			
\safemath{\cohb}{\mu_b}			
\safemath{\mubs}{\nu}	
\safemath{\cohm}{\mu_m} 
\safemath{\dictset}{\setD}	
\safemath{\dictsetp}{\dictset(\coh,\coha,\cohb)}	
\safemath{\dictsetgen}{\dictset_\text{gen}}
\safemath{\dictsetgenp}{\dictsetgen(\coh)}
\safemath{\dictsetonb}{\dictset_\text{onb}}
\safemath{\dictsetonbp}{\dictsetonb(\coh)}

\safemath{\leftside}{U}
\safemath{\rightsideA}{R_a}
\safemath{\rightsideB}{R_b}

\safemath{\indexS}{\setI_S} 

\safemath{\na}{n_a}			
\safemath{\nb}{n_b}			
\safemath{\coeffa}{p_i}	
\safemath{\coeffb}{q_j}	
\safemath{\seta}{\setP}		
\safemath{\setb}{\setQ}     
\safemath{\setw}{\setW}	
\safemath{\setz}{\setZ}	
\safemath{\cola}{\veca}		
\safemath{\colb}{\vecb}		
\safemath{\cold}{\vecd}		
\safemath{\inputvec}{\vecx} 	
\safemath{\error}{\vece}	
\safemath{\noiseout}{\vecz} 	
\safemath{\inputvecel}{x}
\safemath{\inputveca}{\vecx_a}
\safemath{\inputvecb}{\vecx_b}
\safemath{\outputvec}{\vecy}	
\safemath{\lambdamin}{\lambda_{\mathrm{min}}}

\safemath{\elltwo}{\ell_2}
\safemath{\ellone}{\ell_1}
\safemath{\ellzero}{\ell_0}
\safemath{\ellinf}{\ell_\infty}
\safemath{\ellinftilde}{\ell_{\widetilde\infty}}
\safemath{\licard}{Z(\coh,\coha,\cohb)}
\safemath{\xsol}{\hat{x}}
\safemath{\xbord}{x_b}		
\safemath{\xstat}{x_s}		
\safemath{\xstatLone}{\tilde{x}_s}
\safemath{\order}{\mathcal{O}} 
\safemath{\scales}{\Theta} 
\safemath{\ones}{\mathbf{1}} 
\safemath{\zeroes}{\mathbf{0}} 
\safemath{\thlone}{\kappa(\coh,\cohb)} 
\safemath{\constoneA}{\delta} 
\safemath{\constoneB}{\epsilon} 
\safemath{\nlarge}{L}				   
\safemath{\sumlarge}{S_\nlarge}
\safemath{\maxlarger}{P_\nlarge}	   
\safemath{\Pzero}{\textrm{P0}}	
\safemath{\Pone}{\textrm{P1}}
\safemath{\vecfir}{\vecw}			 
\safemath{\vecsec}{\vecz}
\safemath{\elvecfir}{w}              
\safemath{\elvecsec}{z}				 
\safemath{\nlargefir}{n}
\safemath{\normout}{\gamma}
\safemath{\auxfun}{h}
\safemath{\supp}{\textrm{supp}}

\safemath{\indexa}{\ell}
\safemath{\indexb}{r}
\safemath{\indexc}{i}
\safemath{\indexd}{j}

\safemath{\project}{P}

\allowdisplaybreaks 

\safemath{\LAMA}{\textrm{LAMA}}
\safemath{\MRT}{\textrm{MRT}}
\safemath{\betamax}{\beta^\text{max}_\setO}
\safemath{\betamaxno}{\beta^\text{max}}
\safemath{\betamin}{\beta^\text{min}_\setO}
\safemath{\betaminno}{\beta^\text{min}}

\safemath{\Nomin}{\No^\textnormal{min}(\beta)}
\safemath{\Nominnobeta}{\No^\text{min}}
\safemath{\Nomax}{\No^\textnormal{max}(\beta)}
\safemath{\Nomaxnobeta}{\No^\textnormal{max}}
\safemath{\EX}{E_\textnormal{x}}
\safemath{\EXP}{\EX^\textnormal{p}}
\safemath{\Eo}{E_0}

\safemath{\tmax}{{t_\textnormal{max}}}
\safemath{\MAP}{\textrm{MAP}}
\safemath{\IO}{\textrm{IO}}
\safemath{\JO}{\textrm{JO}}
\safemath{\Nopost}{N_{0}^\textnormal{post}}
\safemath{\MT}{U}
\safemath{\MR}{B}
\safemath{\Tran}{\textnormal{T}}
\safemath{\Herm}{\textnormal{H}}
\safemath{\row}{\textnormal{r}}
\safemath{\col}{\textnormal{c}}

\safemath{\NT}{N_\textnormal{T}}
\safemath{\DSNR}{\delta \textnormal{SNR}}
\safemath{\betaMOR}{\beta^{\star}}

\safemath{\Hj}{\bmj}
\safemath{\sj}{w}
\safemath{\Ej}{E_w}
\safemath{\quant}{Q}
\safemath{\compquant}{\mathcal{Q}}
\safemath{\Hest}{\hat{\bH}_{\text{est}}}
\safemath{\EIf}{\boldsymbol{\Lambda}}
\safemath{\rmmse}{\textit{RMSSE}}
\safemath{\Wsnips}{\bW_{\text{SNIPS}}}
\safemath{\Wchops}{\bW_{\text{CHOPS}}}
\safemath{\Pest}{\bP_{\text{est}}}
\safemath{\rproj}{\bmr_{\text{proj}}}
\safemath{\Hproj}{\hat{\bH}_{\text{proj}}}

\begin{document}
	
\title{Jammer Mitigation via Beam-Slicing for \\ Low-Resolution mmWave Massive MU-MIMO}%
\author{\IEEEauthorblockN{Gian Marti$^\ast$, Oscar Casta\~neda$^\ast$, and Christoph Studer}  
\thanks{A conference version of this paper introducing beam-slicing and SNIPS has been presented at the IEEE International Workshop on Signal Processing Systems (SiPS) 2021~\cite{castaneda21sips}. The present manuscript extends our work in~\cite{castaneda21sips} by proposing CHOPS and comparing it to SNIPS,  analyzing the performance of SNIPS and CHOPS for non-line-of-sight channels, and  performing ablation studies that explain design choices when deploying beam-slicing.}
\thanks{$^\ast$GM and OC contributed equally to this work.}
\thanks{GM, OC, and CS are with the Department of Information Technology
and Electrical Engineering, ETH Z\"urich, Switzerland; e-mail: gimarti@ethz.ch, caoscar@ethz.ch, and studer@ethz.ch}
\thanks{The work of OC and CS was supported in part by  ComSenTer, one of six centers in JUMP, a SRC program sponsored by DARPA. The work of CS was also supported by an ETH Research Grant and by the US National Science Foundation (NSF) under grants CNS-1717559 and ECCS-1824379.}
}

\IEEEtitleabstractindextext{\begin{abstract}
Millimeter-wave (mmWave) massive multi-user multiple-input multiple-output (MU-MIMO) promises unprecedented data rates for next-generation wireless systems.
To be practically viable, mmWave massive MU-MIMO basestations (BSs) must rely on low-resolution data converters which leaves them vulnerable to jammer interference.
This paper proposes beam-slicing, a method that mitigates the impact of a permanently transmitting jammer during uplink transmission for BSs equipped with low-resolution analog-to-digital converters (ADCs).
Beam-slicing is a localized analog spatial transform that focuses the jammer energy onto few ADCs, so that the transmitted data can be recovered based on the outputs of the interference-free ADCs.
We demonstrate the efficacy of beam-slicing in combination with two digital jammer-mitigating data detectors: SNIPS and CHOPS.
Soft-Nulling of Interferers with Partitions in Space (SNIPS) combines beam-slicing with a soft-nulling data detector that exploits knowledge of the ADC contamination; projeCtion onto ortHOgonal complement with Partitions in Space (CHOPS) combines beam-slicing with a linear projection that removes all signal components co-linear to an estimate of the jammer channel.
Our results show that beam-slicing enables SNIPS and CHOPS to successfully serve 65\% of the user equipments (UEs) for scenarios in which their antenna-domain counterparts that lack beam-slicing are only able to serve 2\% of the UEs.
\end{abstract}

\begin{IEEEkeywords}
Millimeter wave (mmWave), massive multi-user multiple-input multiple-output (MU-MIMO), analog-to-digital converter (ADC), quantization, jammer mitigation, data detection.
\end{IEEEkeywords}
}

\maketitle
\section{Introduction} 
\label{sec:intro}

\IEEEPARstart{N}{ext}-generation wireless communication systems are expected to rely on the vast, unused bandwidth available at millimeter-wave (mmWave) frequencies in order to meet the ever-growing demand for higher data rates.
Communication at mmWave frequencies is characterized by a high path loss that can be compensated for with massive multiple-input multiple-output (MIMO) technology~\cite{rappaport15a}.
Besides providing the basestation (BS) with a high array gain, massive MIMO also enables multi-user (MU) communication~\cite{larsson14a}.

The deployment of a BS equipped with a large number of antennas and corresponding radio-frequency (RF) chains poses implementation challenges in terms of system costs, power consumption, and circuit complexity.
A potential solution that addresses these challenges is to use low-resolution data converters 
that (i) reduce power consumption of data conversion and (ii) relax the linearity and noise requirements of the RF chains, which, in turn, also translates into power consumption and circuit complexity savings~\cite{jacobsson17b, wang15a}.

Unfortunately, the use of low-resolution analog-to-digital converters (ADCs) leaves the BS vulnerable to jammers that could be introduced, for example, by a rogue user equipment (UE) or a malicious transmitter.
Previous works~\cite{yan14a,shen14a,kapetanovic13a,akhlaghpasand18a,vinogradova16a,zhao17a,do18a,akhlaghpasand20a,akhlaghpasand20b,bagherinejad21a} have analyzed the impact of different types of jamming attacks on massive MU-MIMO systems and proposed mitigation methods based on digital equalization.
However, none of these works take into consideration the compounding challenge of low-resolution data conversion: 
A jammer can either saturate low-resolution ADCs or (if gain-control is used) widen their quantization range, which inevitably drowns the useful signals in quantization noise. Both of these effects introduce distortions that are difficult to remove with subsequent digital processing.  

\begin{figure}[tbp]%
\subfigure[infinite-resolution ADCs]{
\hspace{-3mm}
\includegraphics[width=0.5\columnwidth]{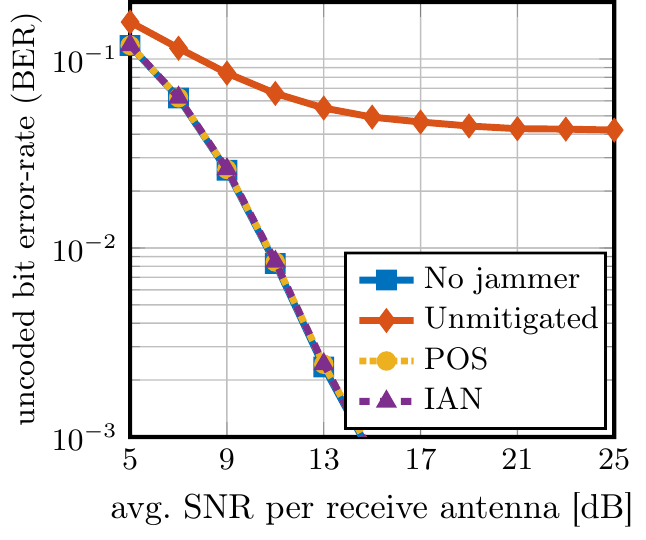}
\label{fig:highres_easy}
}
\hspace{-3mm}
\subfigure[4-bit-resolution ADCs]{
\includegraphics[width=0.5\columnwidth]{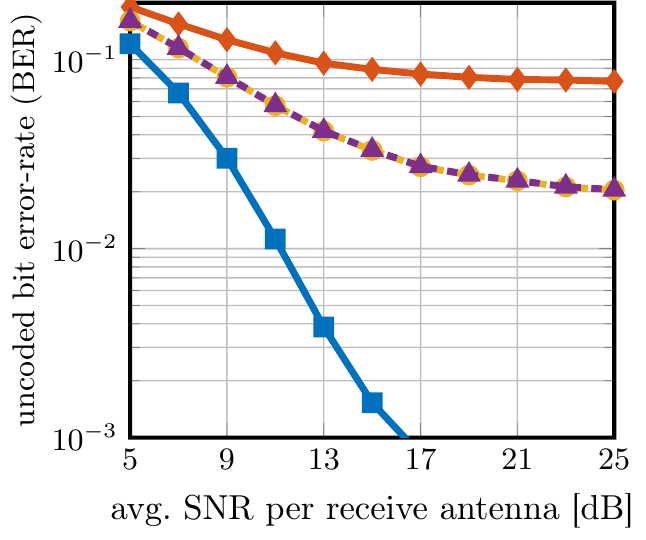}
\label{fig:lowres_hard}
\hspace{-3mm}
}
\caption{Comparing digital jammer mitigation for (a) infinite-resolution ADCs with (b) 4-bit ADCs. 
For infinite-resolution ADCs, the digital methods can remove the jammer with almost no performance loss, but not for low-resolution ADCs. The POS and IAN curves are virtually identical in either case.}
\label{fig:intro_plot}
\end{figure}

For illustration, we consider the efficacy of two classical methods for jammer mitigation: Projection onto the orthogonal subspace (POS) \cite{subbaram1993interference,yan14a}, and linear minimum mean-square error (LMMSE) equalization that treats jammer interference as noise (IAN)~\cite{madhow1994mmse}.
We consider the uncoded bit error-rate (BER) of these mitigation methods both for infinite-resolution ADCs and for 4-bit ADCs. 
Both of these settings consider 32 single-antenna UEs transmitting 16-QAM symbols to a BS equipped with 256 antennas  under line-of-sight (LoS) conditions, and we assume a permanently transmitting jammer with a power 25\,dB stronger than that of the average UE.
In \fref{fig:highres_easy}, we consider infinite-resolution ADCs, where we compare POS and IAN with unmitigated communication, as well as with a baseline case in which no jammer is present.
 All of these methods estimate the channel with least squares (LS) from an orthogonal pilot sequence  and then perform LMMSE equalization. 
In order to mitigate the effects of the jammer, POS projects~both the estimated channel matrix and the receive symbol vector onto the $(B\!-\!1)$-dimensional subspace orthogonal to the (genie-provided) jammer channel before performing LMMSE equalization. 
Contrastingly, IAN uses the (genie-provided) interference covariance matrix to equalize (or \mbox{``soft-null''}) the jammer directly in the LMMSE step, in which the jammer interference is treated as spatially correlated noise.
 We see that, at least when the jammer channel is perfectly known, both POS and IAN achieve almost perfect jammer removal, and the BER performance is virtually identical to the case without jammer.
In contrast, in \fref{fig:lowres_hard}, we consider the behavior when taking into account the quantization artifacts introduced by gain-controlled finite-resolution (4-bit) ADCs.\footnote{See Sections \ref{subsec:adc} and \ref{sec:bussgang} for the quantization procedure, which also uses Bussgang's decomposition to account for the quantization artifacts.} In this setting, both POS and IAN suffer an error floor as high as 2\% BER, even when furnished with \emph{perfect} knowledge of the jammer channel. Since such knowledge would not be obtainable with low-resolution ADCs, the true performance of these jammer mitigation methods is likely worse. The reason for this performance deterioration  is that the ADC inputs are dominated by jammer interference, causing their quantization range to widen up and drown the UE signals in quantization noise, which cannot be removed with linear digital processing.

\subsection{Contributions}
In this work, we develop a practical method that mitigates strong jamming attacks on mmWave massive MU-MIMO systems with low-resolution ADCs at the BS. 
We show that effective jammer mitigation with digital linear equalization is possible only when the ADC resolution is sufficiently high (e.g., 8 bits for a 25\,dB jammer), even when taking into account the nonlinear distortions caused by the ADCs.
However, practical deployments of massive MU-MIMO are likely to rely on \mbox{low-resolution} ADCs,
in which case the jammer will force the ADCs' quantization range to drown the UE signals in quantization noise.
In order to enable jammer-robust communication also with low-resolution ADCs, we propose a novel technique that we call \textit{beam-slicing}: a non-adaptive, localized spatial transform which precedes data conversion.
We then propose two different beam-slicing-based and \mbox{quantization-aware} methods for jammer mitigation, 
Soft-Nulling of Interferers with Partitions in Space (SNIPS) 
and 
projeCtion onto ortHOgonal complement with Partitions in Space (CHOPS), which are essentially the beam-slicing counterparts of IAN and POS from \fref{fig:intro_plot}.
We use simulation results for realistic LoS and non-LoS mmWave channels to demonstrate that beam-slicing enables SNIPS and CHOPS to mitigate the adversarial impact of strong jammers on low-resolution mmWave MU-MIMO systems far more effectively than conventional jammer mitigation schemes that directly operate in antenna-space.
Finally, we justify the design choices in our construction of beam-slicing through ablation studies.

\subsection{Related Prior Work} \label{sec:prior_work}

Several works have studied means to improve the resiliency of MIMO systems against jamming attacks.
These works have considered different attacks, such as constant jamming attacks~\cite{yan14a,shen14a}, in which the jammer is permanently transmitting, as well as other types of attacks in which the jammer transmits only at specific time instances, e.g.,  when the UEs transmit~\cite{yan14a}, or during pilot transmission~\cite{kapetanovic13a}.
Moreover, given the complexity of the jammer problem, some works have devoted themselves only to detecting the presence of a jammer~\cite{kapetanovic13a,akhlaghpasand18a}, while other works have proposed methods to suppress the impact of the detected jammer~\cite{yan14a,shen14a,vinogradova16a,zhao17a,do18a,akhlaghpasand20a,akhlaghpasand20b,bagherinejad21a}.

In this work, we focus on mitigating the interference of a permanently transmitting jammer. 
We now describe existing approaches that deal with such jammer interference.
Reference~\cite{yan14a} proposes a method  for small-scale MIMO systems that uses the angle-of-arrival of the jammer interference to project the receive vector onto its orthogonal subspace.
Also for small-scale MIMO, reference \cite{shen14a} proposes a method that uses differential encoding and exploits the ratio between channel coefficients.
In the context of  massive MIMO, reference \cite{vinogradova16a} uses random matrix theory to estimate the UEs' eigensubspace to then project the received signals onto that subspace, while~\cite{zhao17a} proposes methods which require perfect channel state information and cooperation between the UEs and the BS. 
References \cite{do18a,akhlaghpasand20a} use an estimate of the jammer channel to implement different versions of a jammer-robust zero-forcing detector.

Similarly to our work, references \cite{akhlaghpasand20b,bagherinejad21a} propose to exploit spatially correlated channels to suppress jammer interference.
In particular, the work in~\cite{bagherinejad21a} applies a \textit{beamspace} transform~\cite{brady13}, 
i.e., a spatial discrete Fourier transform (DFT), 
to the BS receive signal in order to separate the jammer from the UEs in angular domain. 
The beamspace transform  is closely related to the concept of beam-slicing proposed in this work.
In fact, the beamspace transform (and even the absence of any spatial transform, i.e., the \textit{antenna domain}) can be formulated as a special case of beam-slicing, and one can think of beam-slicing as a more general beamspace transform with adjustable angular resolution.
The key advantage of beam-slicing over the beamspace transform is that beam-slicing is composed of \textit{localized} transforms that only take inputs from a few adjacent antennas, making it more amenable for analog circuit implementation~\cite{JoshiThesis}.

The difficulty of implementing large analog spatial transforms (specifically, large DFTs) can be illustrated with the example of a Butler matrix.
A Butler matrix is a passive, bidirectional beamforming network consisting of hybrid couplers and fixed phase shifters, and is capable of simultaneously generating multiple beams for antenna arrays~\cite{butler61a,cetinoneri11a}.
To support multibeamforming, the Butler matrix implements a DFT\footnote{While Butler matrices with 180\textdegree~hybrid couplers implement a row-permuted DFT, Butler matrices with 90\textdegree~hybrid couplers implement a rotated DFT that lacks a beam at broadside~\cite{macnamara87a}.} in a structure analogous to that of the fast Fourier transform (FFT)~\cite{macnamara87a,guo21a}.
This efficient topology as well as the reduced component count (compared to the alternative Blass~\cite{blass60a} and Nolen~\cite{nolen65a} matrices) makes the Butler matrix one of the most prominent circuit-based, analog DFT implementations~\cite{guo21a,vallappil21a}.
Nevertheless, the implementation of large Butler matrices remains impractical. Specifically, their implementation is hindered by the more complex routing~\cite{vallappil21a,guo21a} and the necessity for lower manufacturing tolerances~\cite{hall90a}.
Another issue in implementing large Butler matrices is an increase in insertion loss~\cite{guo21a,vallappil21a,tudosie09a}.
These obstacles are reflected by the fact that, to the best of our knowledge, the largest Butler matrices reported in the open literature are for $16$ antenna elements~\cite{klionovski19a,yazdanbakhsh11a,corona03a}.
However, our proposed beam-slicing methods can be implemented with small Butler matrices that transform signals from $4$ or $8$ antennas.

Finally, we note that most existing works on jammer mitigation have not considered the effects of hardware impairments, with the exception of~\cite{akhlaghpasand20b}. 
Reference~\cite{akhlaghpasand20b} models the effects of hardware impairments (including quantization errors) as additive Gaussian noise, hence failing to model the signal- and jammer-dependent distortions introduced by coarse analog-to-digital (A/D) conversion.
In stark contrast, our system simulations explicitly model the effects of low-resolution quantization, whilst our beam-slicing methods, SNIPS and CHOPS, take into consideration such coarse quantization by using Bussgang's decomposition~\cite{bussgang52a,minkoff85a}.

\subsection{Notation}
Matrices and column vectors are represented by boldface uppercase and lowercase letters, respectively.
For a matrix~$\bA$, the conjugate transpose is $\bA^H$, the $k$th column is $\bma_k$, the Frobenius norm is $\| \bA \|_F$, and the trace is $\tr(\bA)$.
The $N\times N$ identity and DFT matrices are $\bI_N$ and $\bF_N$, respectively, where $\herm{\bF_N}\bF_N=\bI_N$.
For a vector $\bma$, the $k$th entry is $a_k$, the $\ell_2$-norm is $\|\bma\|_2$, the real part is $\Re\{\bma\}$, and the imaginary part is $\Im\{\bma\}$.
The all-zeros matrix is denoted by $\mathbf{0}$.
Moreover, $\text{diag}(\bma)$ is a diagonal matrix whose diagonal is formed by the entries of~$\bma$. 
Expectation with respect to the random vector~$\bmx$ is denoted by \Ex{\bmx}{\cdot}.
The floor function $\floor{x}$ returns the greatest integer less than or equal to $x$.
We define~\mbox{$i^2=-1$}.

\section{Propagation Model}\label{sec:system}

We consider the uplink of a mmWave massive MU-MIMO system in which $U$ single-antenna UEs transmit data to a $B$ antenna BS, while a permanently transmitting, single-antenna jammer interferes with the BS receive signal. 
For this scenario, we consider the following frequency-flat 
input-output relation:
 \begin{equation}
\bmy = \bH\bms + \Hj\sj + \bmn. \label{eq:ant_io}
 \end{equation}
Here, $\bmy\in\complexset^B$ is the (unquantized) vector received by the BS antennas, $\bH\in\complexset^{B\times U}$ models the MIMO uplink channel matrix, $\bms\in\setS^U$ is the transmit vector whose (independent) entries correspond to the per-UE transmit symbols which take value in a constellation set $\setS$ (e.g., $16$-QAM), $\Hj\in\complexset^B$ is the channel vector from the jammer to the BS, $\sj\in\complexset$ is the jamming signal, and $\bmn\in\complexset^B$ is i.i.d.\ circularly-symmetric complex Gaussian noise with a per-entry variance of $\No$.
In what follows, we assume that the UE transmit symbols $s_u$, $u=1,\dots,U$, are independent and circularly-symmetric with variance $\Es$ so that $\Ex{\bms}{\bms\herm{\bms}}=\Es\bI_U$.
We model the jamming signal $\sj$ as circularly-symmetric complex Gaussian with variance $\Ej$. 
All probabilistic quantities are assumed to be mutually independent. 

\section{Inteference-Removal with Partitions in Space}
\label{sec:slice}
\begin{figure}[tp]
\centering
\includegraphics[width=\columnwidth]{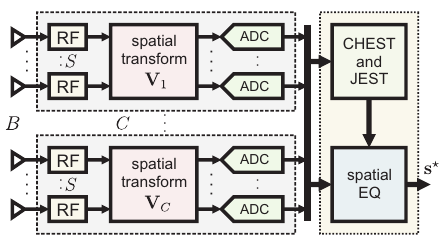}
\caption{System overview of SNIPS and CHOPS: The $B$ RF chains are clustered into $C$ antenna clusters of size $S$. The analog baseband signals are transformed clusterwise to the beam-slice domain before being converted to digital by the ADCs. SNIPS or CHOPS then performs jammer interference estimation, channel estimation, and data detection in the beam-slice domain.}
\label{fig:system_overview}
\end{figure}

Our approach aims to protect most of the ADCs from the jammer by exploiting the strong spatial directivity of mmWave signals. 
Prior to A/D-conversion, we apply a spatial transform that resolves the incident waves so that only a few ADCs are strongly affected by the jammer.
One can then discount the outputs of these jammer-distorted ADCs during equalization and detect the data symbols mainly based on the outputs of the distortion-free ADCs.

A na\"ive approach would be to transform the (unquantized) receive vector~$\bmy$ to the beamspace (or angular) domain~\cite{brady13} using a DFT according to $\bmy_{B}=\bF_B \bmy$, and set the entries of the beamspace vector $\bmy_{B}$ dominated by jammer interference to zero.
Accordingly, the resulting vector $\bmy_{B,\text{mask}}$ will have entries equal to zero if they belong to jammer-contaminated beams, and otherwise equal to the corresponding entry in $\bmy_{B}$.
One could then equalize $\bmy_{B,\text{mask}}$ as if neither interference nor interference-cancellation had occurred, for instance with LMMSE estimation, 
\mbox{$\bms^\star = \inv{(\herm{\bH}\bH + \No/\Es\bI_B)}\herm{\bH}\bmy_{B\text{,mask}}$.}
While such an approach would be effective in suppressing jammer interference, implementing large analog spatial transforms, such as the DFT, is nontrivial (cf.~\fref{sec:prior_work}), especially when considering hundreds of BS antennas~\cite{JoshiThesis,vallappil21a,guo21a,hall90a,tudosie09a}.

As a practical alternative, we propose beam-slicing, a distributed, localized, and hence small analog spatial transform that can be implemented in practice. We also propose two jammer mitigation methods based on beam-slicing, SNIPS and CHOPS.
SNIPS and CHOPS do not discard the outputs of jammer-affected ADCs completely, but instead take into account each ADC-output's fidelity by estimating the amount of jammer interference at the individual ADCs.
Moreover, SNIPS and CHOPS utilize Bussgang's decomposition in order to take into account the effects of practical, low-resolution ADCs.
The only difference between SNIPS and CHOPS lies in how they mitigate jammer interference: SNIPS performs LMMSE equalization where the jammer interference is treated as noise; CHOPS projects the receive signal (and channel estimate) onto the subspace orthogonal to the jammer channel before applying a conventional LMMSE equalizer. 

In the following discussion of these two methods, SNIPS and CHOPS differ only in Section \ref{subsec:data_detection}. The rest of their jammer-mitigation pipeline, illustrated by \fref{fig:system_overview}, is identical.

\begin{figure}[tp]
\centering
\includegraphics[width=0.95\columnwidth]{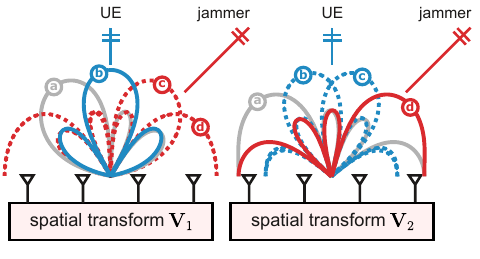}
\caption{Effect of $S\!=\!4$ beam-slicing for a $B\!=\!8$ antenna array considering one UE and one jammer (both in far-field). Spatial transform $\bV_1$ illustrates the $4$ beams (a), (b), (c), and (d) of $\bF_4$, of which (b) is perfectly aligned to capture the UE's transmitted power, while beams (c) and (d) partially capture the jammer's power. Such partial capturing would lead to unsatisfactory jammer mitigation. To increase angular diversity, the spatial transform $\bV_2$ first rotates the received signal so that the $4$ beams (a), (b), (c), and (d)  of the resulting transform are also shifted. As a result, the shifted beam (d) is now able to capture the jammer's power, which allows for better jammer mitigation.}
\label{fig:explain_beamslicing}
\end{figure}

\subsection{Beam-Slicing} \label{subsec:beam_slicing}

Beam-slicing transforms partitions of the (unquantized) receive vector~$\bmy$ into a ``rotated'' angular domain.
Beam-slicing is fully analog, non-adaptive, and operates in decentralized fashion. Specifically, we partition the BS antenna array into~$C$ clusters of equal size, each consisting of $S=B/C$ adjacent BS antennas. The corresponding partition of the receive vector is denoted $\bmy = \tp{[\tp{\bmy_1}, \dots, \tp{\bmy_C}]}$, where $\bmy_c\in\complexset^S, c=1,\dots,C$.
Beam-slicing then transforms the receive clusters into what we call the \textit{beam-slice domain} as follows:
\begin{equation}
	\hat{\bmy}_c = \bV_c\,\bmy_c, \quad c=1,\dots,C.
\end{equation}
Here, the $c$th cluster matrix $ \bV_c$ is given as a progressively phase-shifted $S$-point DFT matrix $\bF_S$ according to 
\begin{align}
	\bV_c = 
	\!\bF_S~\text{diag}\!\left(1, 
	        \dots, e^{-i\frac{2\pi}{B}(c-1)(s-1)},
	\dots, e^{-i\frac{2\pi}{B}(c-1)(S-1)}
	\right)\!, \label{eq:Vc} \end{align}
so that the different clusters transform to successively ``rotated'' angular domains. 	
Such phase-rotated DFTs are used to increase the ``angular diversity'' of beam-slicing to better capture the possible directions of jammers---see \fref{fig:explain_beamslicing} for a graphical explanation, as well as \fref{sec:ablation} for empirical evaluation.
The action of the beam-slicer is summarized~as
\begin{align}
\hat{\bmy} = \bV\bmy = \begin{bmatrix}
	\bV_1\bmy_1 \\ \vdots \\ \bV_C\bmy_C
\end{bmatrix},	
\end{align}
where $\bV = \text{diag}(\bV_1, \dots, \bV_C)$ and $\herm{\bV}\bV=\bI_B$.
We also point out that for an (impractical) cluster size $S=B$, beam-slicing corresponds to performing a conventional beamspace transform.
In what follows, it will be convenient to 
define the \emph{beam-sliced channel matrix} $\hat\bH = \bV\bH$, and the \emph{beam-sliced jammer channel} $\hat\Hj = \bV\Hj$, which allows us to rewrite \eqref{eq:ant_io} as
\begin{align}
	\hat{\bmy} = \hat\bH \bms + \hat\Hj\sj + \hat\bmn, \label{eq:io_beamspace}
\end{align}
where $\hat\bmn=\bV\bmn$ has the same distribution as $\bmn$. 

\subsection{Data Conversion} \label{subsec:adc}
The beam-sliced signal $\hat{\bmy}$ is then converted into the digital domain.
To~take into account the quantization errors of low-resolution ADCs, we assume that the beam-sliced vector~$\hat{\bmy}$ is quantized as follows:
\begin{align}
	\bmr = \inv{\bG}\left(\quant\left(\Re\{\bG \hat{\bmy}\}\right) + i\quant\left(\Im\{\bG \hat{\bmy}\}\right)\right). \label{eq:quantization1}
\end{align}
Here, $\bG=\text{diag}(g_1,\dots,g_B)$ is a diagonal matrix that represents beam-wise gain-control. 
The quantization function $\quant(\cdot)$ is applied entry-wise to its input and represents a $q$-bit uniform midrise quantizer with step size $\Delta$ defined as
\begin{align} \label{eq:quantizer}
\quant(x)\triangleq
    \begin{cases}
      \Delta\floor{\frac{x}{\Delta}}+\frac{\Delta}{2}, & \text{if}\ |x|<\Delta 2^{q-1}\\
      \frac{\Delta}{2}(2^q-1)\frac{x}{|x|}, & \text{if}\ |x|>\Delta 2^{q-1}.
    \end{cases}
\end{align}
For the quantizer's step size $\Delta$, we use the value which minimizes the mean-square error (MSE) between the quantizer's output~$\quant(x)$ and its input~$x$ under the assumption that~$x$ is Gaussian with zero mean and unit variance~\cite{max60a}.
For convenience, we will denote \eqref{eq:quantization1} as
\begin{align}
	\bmr = \compquant(\hat\bmy). \label{eq:quantization2}
\end{align}

The per-beam gains aim to ensure that the values entering the quantizers have unit variance per real dimension and are obtained from a set of $T$ training vectors 
$\tilde\bY=[\tilde\bmy_1, \dots, \tilde\bmy_T]$ as 
\begin{align}
	g_b = \sqrt{\frac{2T}{\|\tilde\bmy_{(b)}\|_2^2}},
	\quad b=1,\dots,B, \label{eq:learnG}
\end{align}
where $\tilde\bmy_{(b)}$ is the $b$th row of $\tilde\bY$.  

\subsection{Jammer Interference Estimation} \label{subsec:jest}
Our interference-mitigating data detection schemes (see below) rely on knowledge of statistics of the jammer's interference. 
Specifically, these methods need to know the covariance matrix $\bC_j = \Ex{\sj}{\hat\Hj \sj\herm{(\hat\Hj \sj)}}$. 
We suggest to estimate $\bC_j$ from a number of channel uses during which the UEs are not transmitting and where the jammer transmits i.i.d.\ jamming symbols $[\sj_1, \dots, \sj_N]$, so the beam-sliced receive matrix and corresponding quantization output can be modeled as
\begin{align}
	\hat\bY_J &= \hat\bmj [\sj_1, \dots, \sj_N] + \hat\bN \\
	\bR_J &= \compquant(\hat\bY_J). \label{eq:jammer_quantiz}
\end{align}
In order to learn the jammer channel, we propose to estimate the gain matrix $\bG$ with \eqref{eq:learnG} directly from the received signals, $\bG=\bG(\hat\bY_J)$.
Our estimate $\EIf$ of the covariance matrix $\bC_j$ is given by
\begin{align} \label{eq:covarianceestimate}
	\EIf = \frac1N \bR_J\herm{\bR_J}. 
\end{align}

\subsection{Channel Estimation} \label{subsec:chest}
We estimate the UEs' channel matrix using a pilot-based LS estimator from $U$ orthogonal pilot sequences \mbox{$\bS_P = [\bms_1, \dots, \bms_U]$.} 
The channel estimation pipeline passes through the beam-slicer and the quantizer. The beam-sliced receive matrix and corresponding quantization are modeled as 
\begin{align}
	\hat\bY_P &= \hat\bH \bS_P + \hat\Hj[w_1, \dots, w_U] + \hat\bN \\
	\bR_P &= \compquant(\hat\bY_P), \label{eq:pilot_receive}
\end{align}
where we estimate the gain matrix $\bG$ with \eqref{eq:learnG} from the pilot sequence itself, $\tilde\bY=\hat\bY_P$.
(We fix this choice of $\bG=\bG(\hat\bY_P)$ also for the data detection phase described below.) 
We then estimate the beam-sliced channel matrix with an LS estimate: 
\begin{align}
	\Hest &= \bR_P \herm{\bS_P}\inv{\left(\bS_P\herm{\bS_P}\right)} \label{eq:pilot1} \\
	&\overset{(a)}{=} \frac{1}{U\Es} \bR_P \herm{\bS_P}, \label{eq:pilot2} 
\end{align}
where $(a)$ holds because the pilot sequence is orthogonal.\footnote{This LS channel estimator ignores the jammer interference. In this regard, see also the explanatory note in Footnote \ref{epicfootnotereference}.}

\subsection{Bussgang Analysis}
\label{sec:bussgang}
So far (i.e., for jammer covariance estimation and channel estimation), we have neglected the distortion introduced by the quantization step in \eqref{eq:quantization1}--\eqref{eq:quantization2}. 
We do not, however, neglect this distortion during the data detection step:\footnote{In our experiments, we also tried applying the Bussgang decomposition from \fref{sec:bussgang} to channel estimation. However, we did not observe noticeable error-rate performance improvements. For this reason, our methods SNIPS and CHOPS apply Bussgang decomposition only for data detection.}  
The quantization step introduces distortions which are correlated with the quantizer inputs. We assume that the components of the quantizer inputs are real-valued Gaussian with zero mean and unit variance. This assumption allows us to perform a component-wise Bussgang decomposition~\cite{bussgang52a} of the quantization signal as follows:
\begin{equation}
	\quant(x) = \gamma \, x + d.
\end{equation}
Here, $\gamma$ is the quantizer's \emph{Bussgang gain}, and the distortion~$d$ has zero mean and is uncorrelated with $x$. The Bussgang gain is given by \cite[Eq.\,(9)]{minkoff85a}
\begin{equation}
	\gamma = \frac{\Ex{}{\quant(x)x}}{\Ex{}{x^2}}
\end{equation}
and the variance of the distortion $d$ is
\begin{equation}
	D = \Ex{}{d^2} = \Ex{}{\quant(x)^2} - \gamma^2 \Ex{}{x^2}. \label{eq:bussgang_distortion}
\end{equation}
Bussgang's decomposition allows us to rewrite~\eqref{eq:quantization2} as 
\begin{align}
	\bmr &= \compquant(\hat\bmy) \\
	&= \inv{\bG}\left(\quant\left(\Re\{\bG \hat{\bmy}\}\right) + i\quant\left(\Im\{\bG \hat{\bmy}\}\right)\right)\\
	&= \inv{\bG}\left(\gamma\,\Re\{\bG \hat{\bmy}\} + \bmd_r + i\left(\gamma\,\Im\{\bG \hat{\bmy}\} + \bmd_i\right)\right) \\
	&= \gamma\, \inv{\bG} \left(\Re\{\bG \hat{\bmy}\} + i\Im\{\bG \hat{\bmy}\}\right) +\inv{\bG}( \bmd_r + i\bmd_i) \\
	&= \gamma\, \hat{\bmy} + \inv{\bG} \bmd, \label{eq:bussgang_almostfinal}
\end{align}	
where we define $\bmd = \bmd_r + i\bmd_i$. Based on \eqref{eq:bussgang_distortion}, we make the idealized assumption that the covariance matrix of $\bmd$ is
\begin{equation}
	\bC_d = \Ex{\bmd}{\bmd\herm{\bmd}} \approx 2D\,\bI_B. \label{eq:assump_diag}
\end{equation}

We now combine \eqref{eq:bussgang_almostfinal} with \eqref{eq:io_beamspace}, which results in the following input-output relation:
\begin{equation}
	\bmr = \gamma\hat\bH\bms + \gamma\hat\Hj\sj + \gamma\hat\bmn + \inv{\bG} \bmd. \label{eq:bussgang_final}
\end{equation}

The expression in \eqref{eq:bussgang_final} also sheds light on how low-resolution ADCs exacerbate the jammer interference: If there is significant jamming energy at the ADC inputs, then the elements $g_b$ in~\fref{eq:learnG} of the (diagonal) gain control matrix $\bG$ are close to zero. Therefore, the entries of $\inv{\bG}$ are very large, and we see from the last term in \eqref{eq:bussgang_final} how this leads to amplification of the quantization noise $\bmd$.
Beam-slicing aims to ensure that only a small set  of the entries of $\inv{\bG}$ are large, so that the signal can be detected from the remaining components of $\bmr$.

\subsection{Interference-Removing Data Detection: SNIPS \& CHOPS} \label{subsec:data_detection}

We now describe two different ways of treating the jammer interference in the beam-sliced receive vector $\bmr$ from \eqref{eq:bussgang_final} by means of linear equalization, resulting in SNIPS and CHOPS. 

\emph{1) SNIPS:}
The first method uses an LMMSE-like detector that treats the jammer interference as noise, resulting in SNIPS. In deriving the SNIPS data detector, we make certain idealizing assumptions. 
The detector is
\begin{align}
	\bms^\star = \Wsnips\,\bmr, \label{eq:snips}
\end{align}
where the matrix $\Wsnips$ is given by
 \begin{align}
 	&\Wsnips \nonumber\\
 	&= \gamma\Es\herm{\Hest}\inv{\big( \gamma^2\Es\Hest\herm{\Hest} \!+ \gamma^2\EIf \!+ \gamma^2\No\bI_B \!+ 2D\bG^{-2}\big)} \\
 	&= \frac{1}{\gamma} \herm{\Hest}\inv{\Big( \Hest\herm{\Hest} + 
 		\frac{1}{\Es}(\EIf + \No \bI_B + 2D\gamma^{-2}\bG^{-2})\Big)}. \label{eq:WF1}
 \end{align}
Here, we use the gain control matrix acquired during the pilot phase, $\bG=\bG(\hat\bY_P)$.
If the diagonal approximation \eqref{eq:assump_diag} and the approximations $\Hest\approx\hat\bH$ and $\EIf \approx \bC_j$ were exact, equation \eqref{eq:snips} would implement the LMMSE estimator when the interference is regarded as noise.  

\emph{2) CHOPS:}
The second method projects the received signals onto the $(B-1)$-dimensional subspace orthogonal to the beam-sliced jammer channel and performs LMMSE-like data detection in this projection space. 
The matrix for this projection would be \mbox{$\bP = \bI_B - \hat\Hj \herm{\hat\Hj} / \|\hat\Hj\|^2$} \cite[Sec.\,2.6.1]{GV96}, which can be approximated based on the covariance estimate $\EIf$ in \fref{eq:covarianceestimate} as
\begin{align}
	\Pest = \bI_B - \frac{\EIf}{\tr{(\EIf)}}; \label{eq:pest}
\end{align}
see Appendix \ref{app:projection} for a derivation.
To obtain a consistent model of the wireless channel, the projection is also applied during channel estimation. So, for CHOPS, we replace \eqref{eq:pilot_receive} with 
\begin{align}
	\bR_P &= \Pest\compquant(\hat\bY_P), \label{eq:chop_pilot_receive}
\end{align}
while the rest of the channel estimation phase remains identical as in Section \ref{subsec:chest}. 
We denote the obtained estimate of the projected, beam-sliced channel matrix with $\Hproj$. 
(Mathematically, instead of projecting the quantized pilot receive matrix on the orthogonal subspace as in \eqref{eq:chop_pilot_receive}, 
we could equivalently have projected the estimate of the beam-sliced channel matrix \eqref{eq:pilot_receive}, as $\Hproj = \Pest\, \Hest$.\footnote{\label{epicfootnotereference}Note that we did not explicitly address the jammer contamination of the channel estimate in SNIPS. The reason is that we observe empirically that the malicious influence of this contamination is negated by the $\EIf$-term in the inverse of \eqref{eq:WF1}, which makes such a projection unnecessary. To illustrate this, we point out that the IAN-curve in \fref{fig:highres_easy} is also based on a detector whose channel estimate suffers from strong jammer contamination: Its performance nevertheless matches the performance of the case without jammer (where the channel estimate is uncontaminated), showing that this channel contamination ultimately has no consequences on the error-rate performance.
})
For data detection, the beam-sliced receive vector $\bmr$ from \eqref{eq:bussgang_final} is  projected accordingly: 
\begin{align}
	\rproj = \Pest\, \bmr.
\end{align}
We then perform LMMSE-like data detection on the resulting quantities as 
\begin{align}
	\bms^\star = \Wchops\,\rproj \label{eq:chops}
\end{align}
with
\begin{align}
	&\Wchops \nonumber\\ 
	&= \frac{1}{\gamma} \herm{\Hproj}\inv{\Big( \Hproj\herm{\Hproj} + 
 		\frac{1}{\Es}(\No \bI_B + 2D\gamma^{-2}\bG^{-2})\Big)},
\end{align}
which would implement LMMSE detection in the projection space if the diagonal approximation \eqref{eq:assump_diag} as well as the~approximations $\Hest\approx\hat\bH$ and $\EIf \approx \bC_j$ were exact.

\section{Results}
\label{sec:res}

\begin{figure*}[t]%
\centering 
\subfigure[LoS]{
\includegraphics[width=0.8\columnwidth]{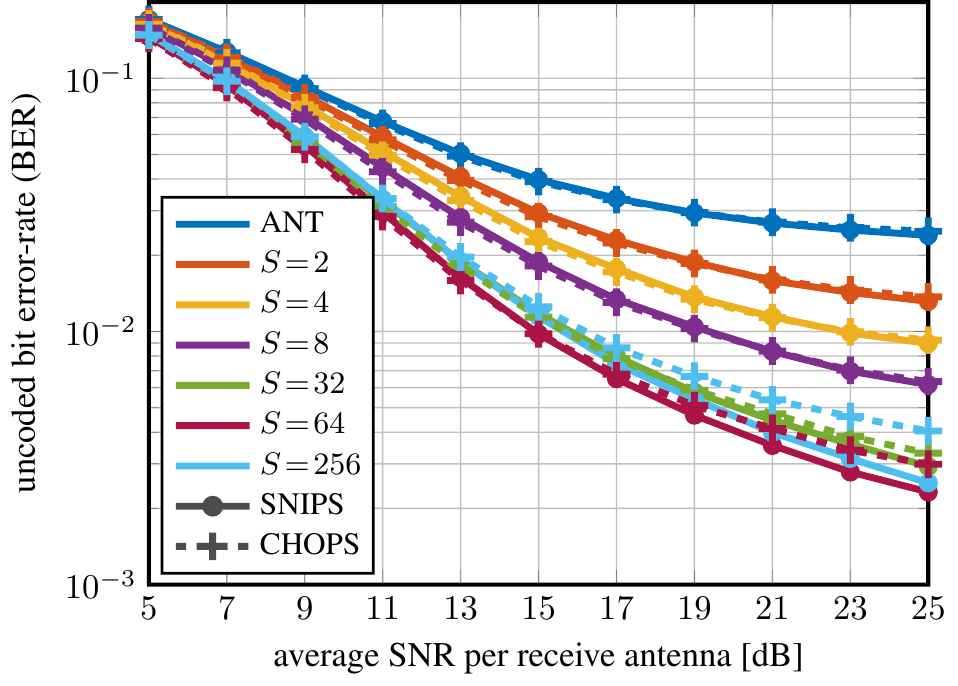}
\label{fig:sweep_S:ber:los}
}
~~\quad\quad\quad
\subfigure[Non-LoS]{
\includegraphics[width=0.8\columnwidth]{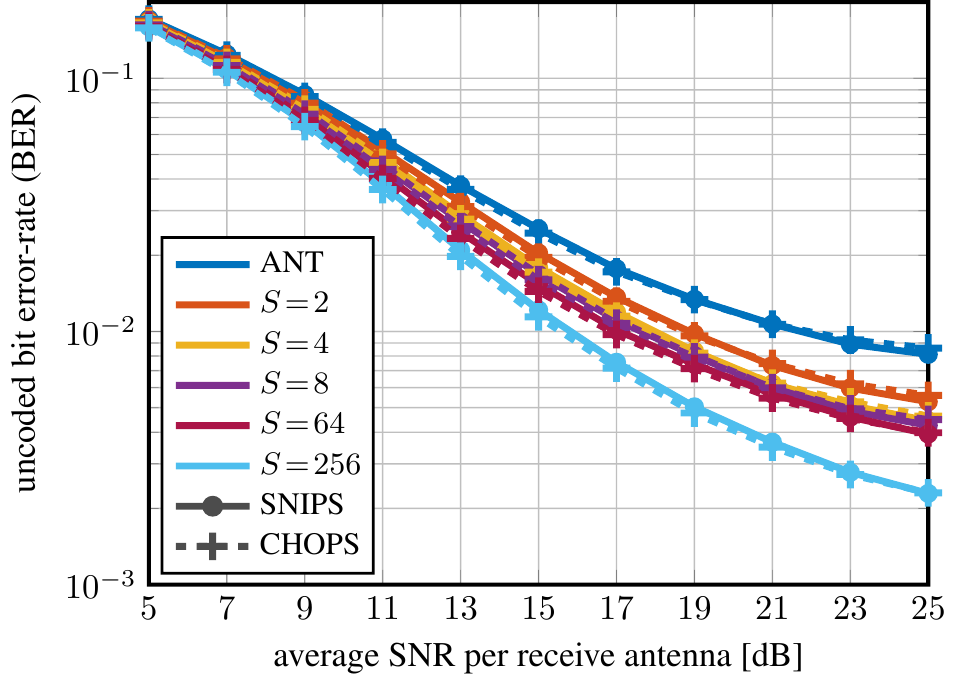}
\label{fig:sweep_S:ber:nlos}
}
\caption{Uncoded bit error-rate (BER) when performing jammer mitigation in the antenna (ANT) and beam-slicing domains with SNIPS (solid) and CHOPS (dashed) for different cluster sizes $S$. We consider transmission over both (a) LoS and (b) non-LoS channels. The relative jammer power is $\rho=25$\,dB and the ADC resolution is $q=4$ bits.}
\label{fig:sweep_S:ber}

\subfigure[LoS]{
\includegraphics[width=0.8\columnwidth]{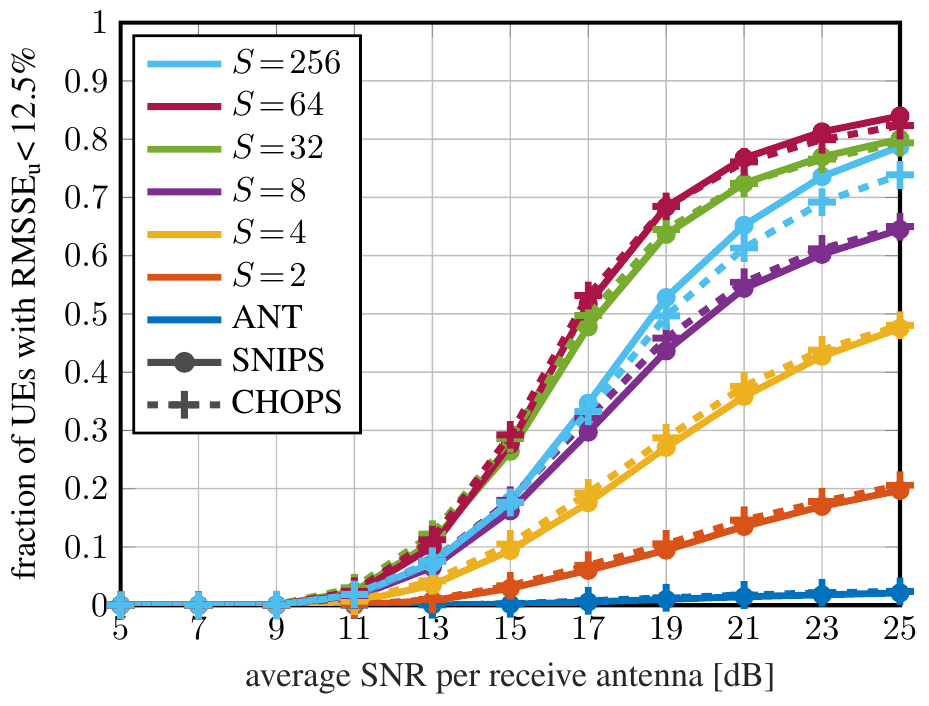}
\label{fig:sweep_S:cdf:los}
}
~~\quad\quad\quad 
\subfigure[Non-LoS]{
\includegraphics[width=0.8\columnwidth]{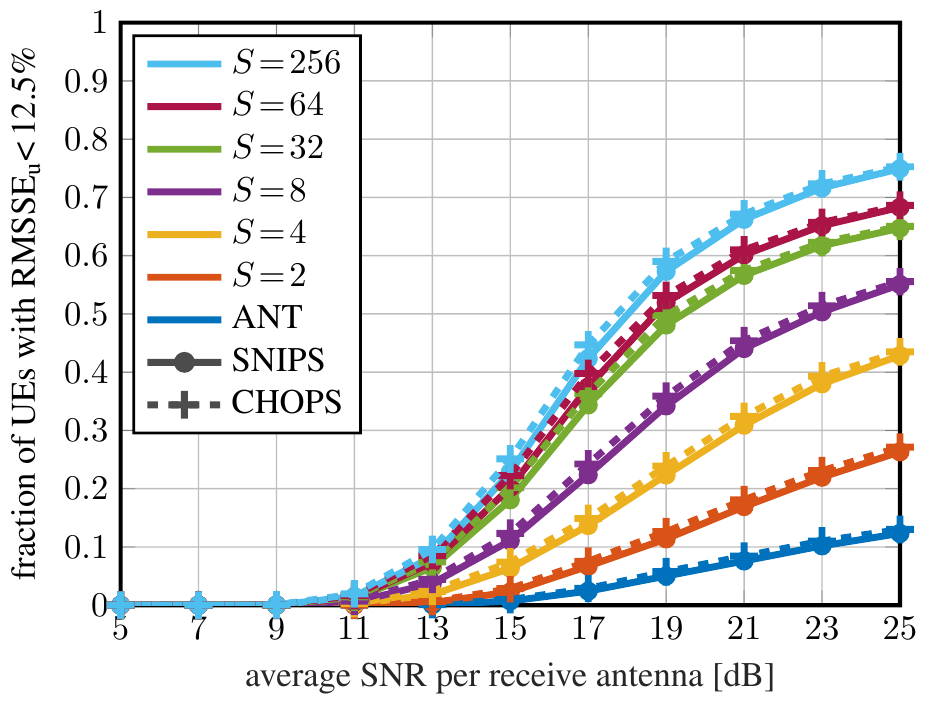}
\label{fig:sweep_S:cdf:nlos}
}
\caption{Fraction of successfully served UEs when performing jammer mitigation in the antenna (ANT) and beam-slicing domains with SNIPS (solid) and CHOPS (dashed) for different cluster sizes $S$. We consider transmission over both (a) LoS and (b) non-LoS channels.  The relative jammer power is $\rho=25$\,dB and the ADC resolution is $q=4$ bits.}
\label{fig:sweep_S:cdf}

\end{figure*}

We now demonstrate the efficacy of beam-slicing by comparing SNIPS and CHOPS with two baselines that differ from SNIPS and CHOPS only in lacking analog beam-slicing.
We note that the operations of these baselines correspond to SNIPS and CHOPS with cluster size $S=1$, which implies $\bV=\bI_B$:
The baselines perform A/D-conversion and soft-nulling or orthogonal projection directly in the antenna domain\footnote{We note that the baselines consisting of SNIPS and CHOPS in the antenna domain correspond to IAN and POS from \fref{fig:intro_plot}, respectively. However, for simplicity, we refrain from using the terms IAN and POS in the remainder.}.  
We will show that in the presence of a strong jammer, beam-slicing with a two-antenna cluster size ($S=2$) already yields significant improvements over these baselines.
We will also identify situations that benefit from beam-slicing considering different ADC resolutions and levels of jamming power.
Finally, we will empirically justify two of our choices when designing beam-slicing, namely the use of the DFT as the base spatial transform and the use of ``rotated'' DFTs for each cluster.

\subsection{Simulation Setup and Performance Metrics}
We simulate a mmWave massive MIMO system in which $U=32$ single-antenna UEs communicate to a $B=256$ antenna BS both under LoS and non-LoS conditions.
The UE and jammer channels are generated using the QuaDRiGa mmMAGIC urban microcellular (UMi) model~\cite{jaeckel2014quadriga} for a carrier frequency of $60$\,GHz and a uniform linear array (ULA) with half-wavelength spacing. We let the $U$ UEs and the jammer be randomly placed at distances from $10$\,m to $100$\,m within a $120$\textdegree~angular sector in front of the BS. The minimum angular separation between two UEs, as well as between the jammer and any UE, is $1$\textdegree. We assume $\pm3\,$dB per-UE power control, so that the ratio between maximum and minimum per-UE receive power is $4$. The transmit constellation is 16-QAM. 
In our simulations, we define the average receive signal-to-noise ratio (SNR) as 
\begin{align}
\textit{SNR} \define \frac{\Ex{\bms}{\|\bH\bms\|_2^2}}{\Ex{\bmn}{\|\bmn\|_2^2}}.
\end{align}
To quantify the jammer's power in comparison to a single UE, we define the relative jammer power $\rho$ as
\begin{align}
\rho \define \frac{U\Ex{\sj}{\|\Hj\sj\|_2^2}}{\Ex{\bms}{\|\bH\bms\|_2^2}} =  \frac{U\Ej \|\Hj\|_2^2}{\Es\|\bH\|_F^2}.
\end{align}

We will consider two performance metrics: Uncoded BER and the per-UE root mean-square symbol error (RMSSE)~\cite{Song21}. The RMSSE for the $u$th UE over $n$ data symbol slots is: 
\begin{align}
	\rmmse_u \triangleq \sqrt{ \frac{\sum_{k=1}^{n}\big|s^\star_{u,k} - s_{u,k} \big|^2 }
	{\sum_{k=1}^{n} \left|s_{u,k} \right|^2} }. \label{eq:rmmse}
\end{align}
Here, $s_{u,k}$ and $s^\star_{u,k}$ are the transmitted and estimated data symbol of the $u$th UE at time slot $k$, respectively.
To understand the relevance of $\rmmse_u$ as a performance metric, it is helpful to compare it to the error vector magnitude (EVM) requirements in the 3GPP 5G NR technical specification~\cite{3gpp21a}. The EVM is loosely speaking the square root of the sum of $\rmmse_u$-squared over all $U$ UEs. Vice versa, the $\rmmse_u$ is loosely speaking a single-UE proxy for the EVM. We will therefore interpret $\rmmse_u$ as a random variable and analyze its distribution by means of Monte-Carlo simulations. 
For 16-QAM transmission, the 3GPP 5G NR technical specification requires an EVM below $12.5$\%~\cite[Tbl. 6.5.2.2-1]{3gpp21a}. As a performance metric, we therefore consider the fraction of UEs (averaged over many UE placements/channel realizations, noise and jammer realizations, and data transmissions) for which the $\rmmse_u$ is below $12.5$\%. 

\begin{figure*}[tp]
\centering
\subfigure[CHOPS, LoS]{
\includegraphics[width=0.31\linewidth]{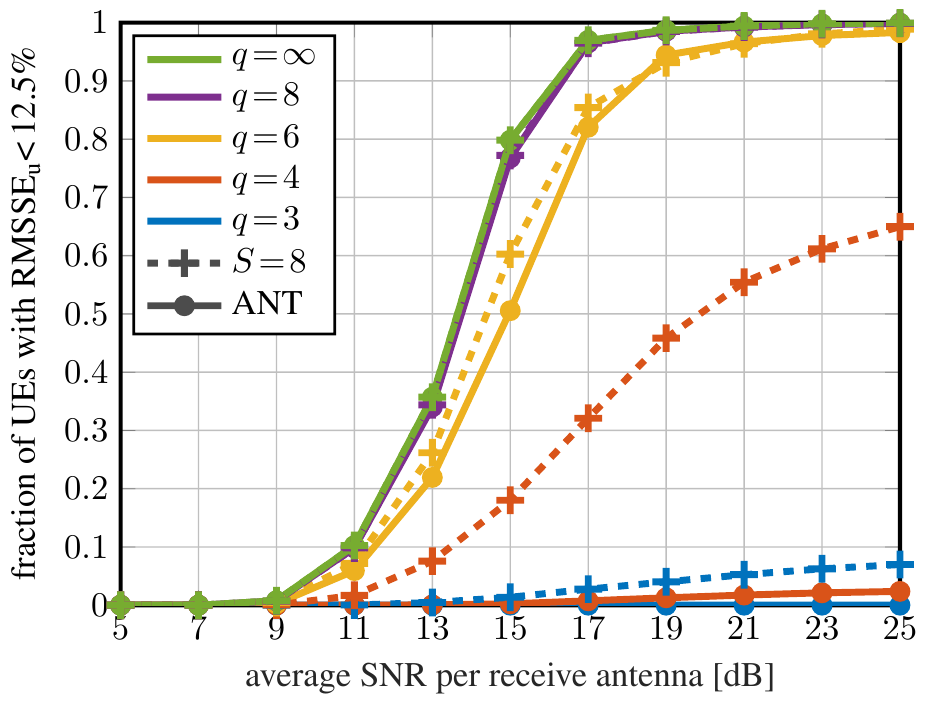}
\label{fig:sweep_ADC:pos_los}
}
\subfigure[SNIPS, LoS]{
\includegraphics[width=0.31\linewidth]{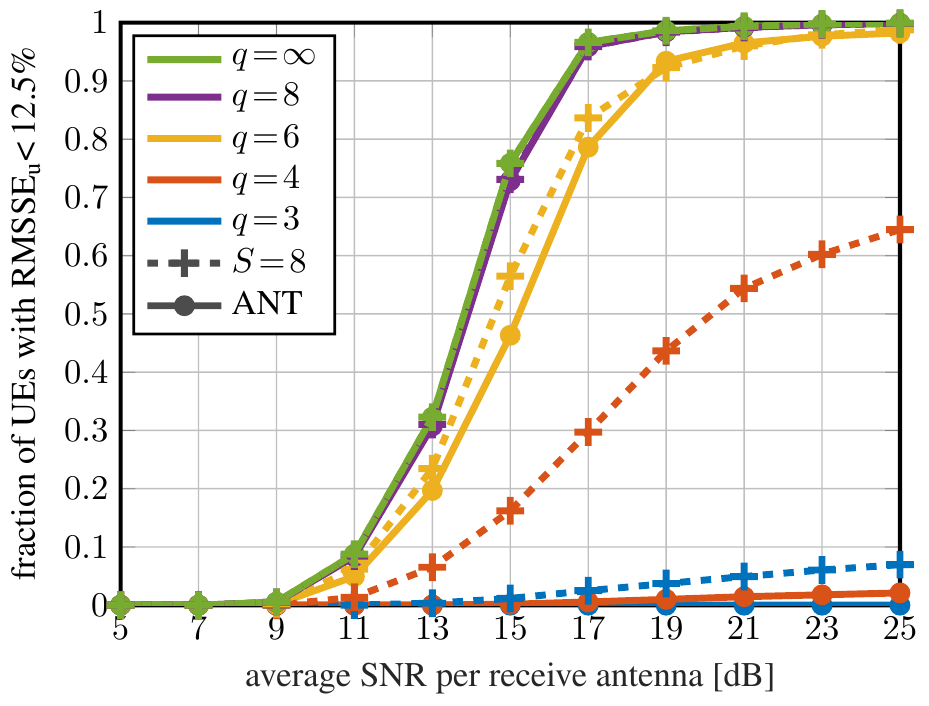}
\label{fig:sweep_ADC:lmmse_los}
}
\subfigure[SNIPS, non-LoS]{
\includegraphics[width=0.31\linewidth]{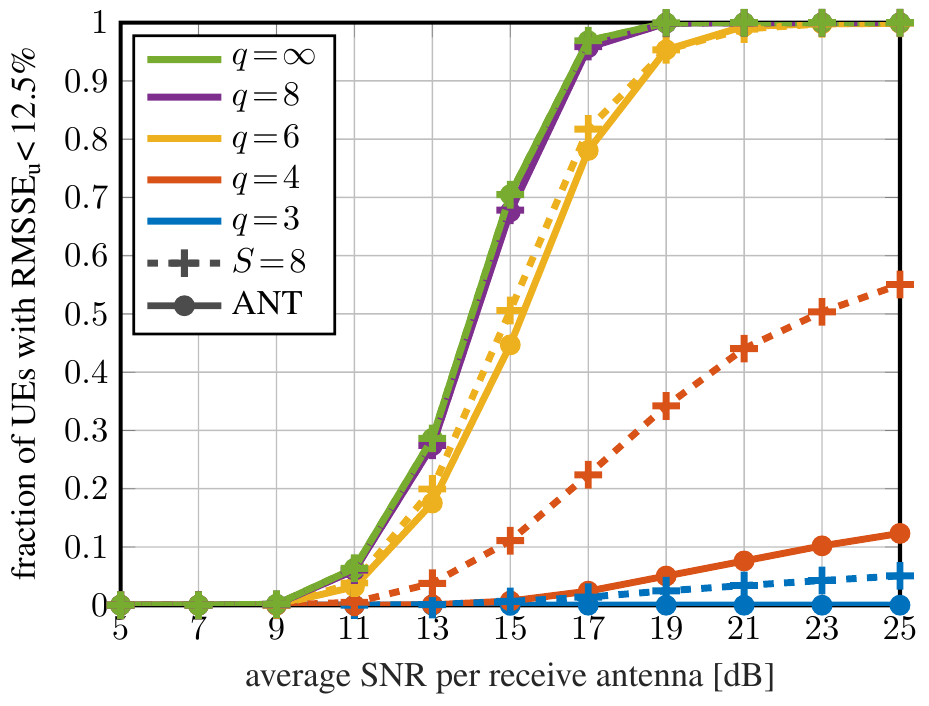}
\label{fig:sweep_ADC:lmmse_nlos}
}
\caption{Fraction of successfully served UEs vs.~SNR, compared between antenna domain (ANT; solid curves) and beam-slicing with $S=8$ (dashed curves) for different ADC resolutions of $q$ bits per real dimension. The relative jammer power is $\rho=25$\,dB.}
\label{fig:sweep_ADC}
\end{figure*}
\begin{figure*}[tp]
\centering
\subfigure[CHOPS, LoS]{
\includegraphics[width=0.31\linewidth]{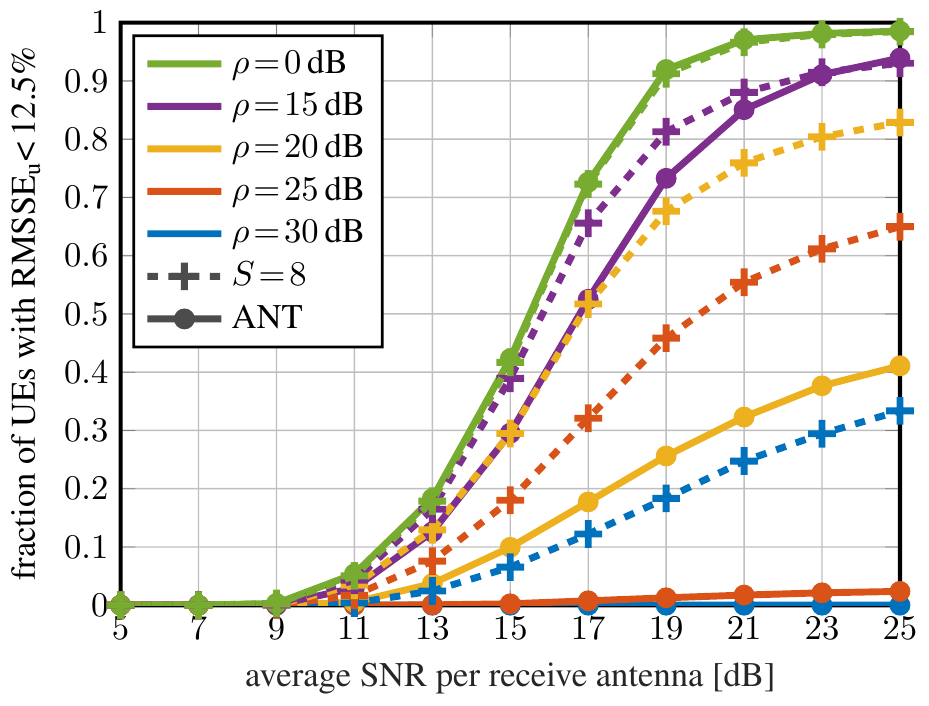}
\label{fig:sweep_jammer:pos_los}
}
\subfigure[SNIPS, LoS]{
\includegraphics[width=0.31\linewidth]{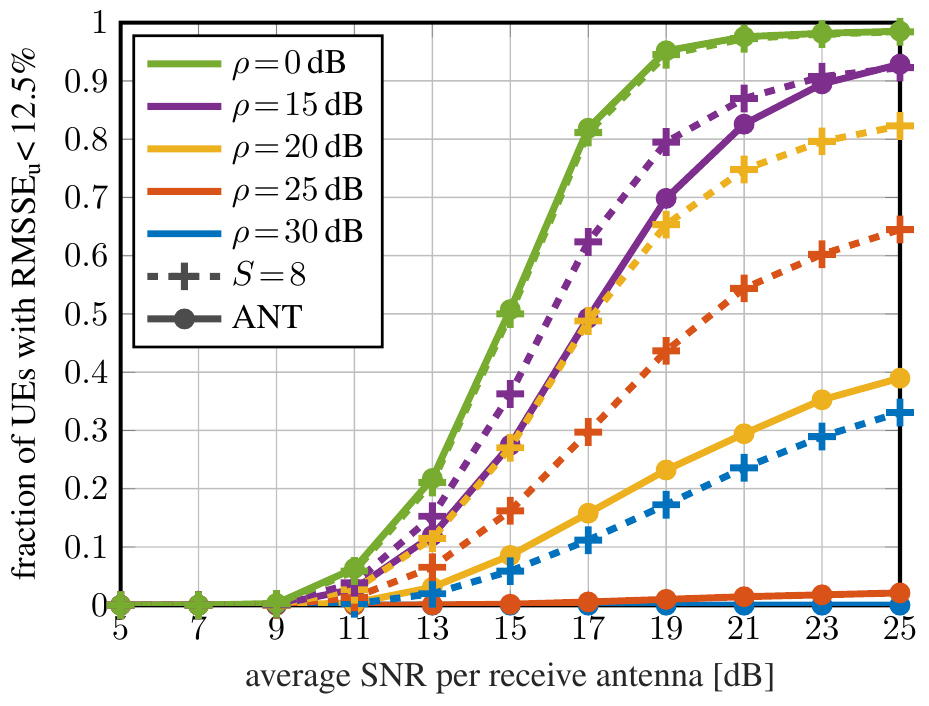}
\label{fig:sweep_jammer:lmmse_los}
}
\subfigure[SNIPS, non-LoS]{
\includegraphics[width=0.31\linewidth]{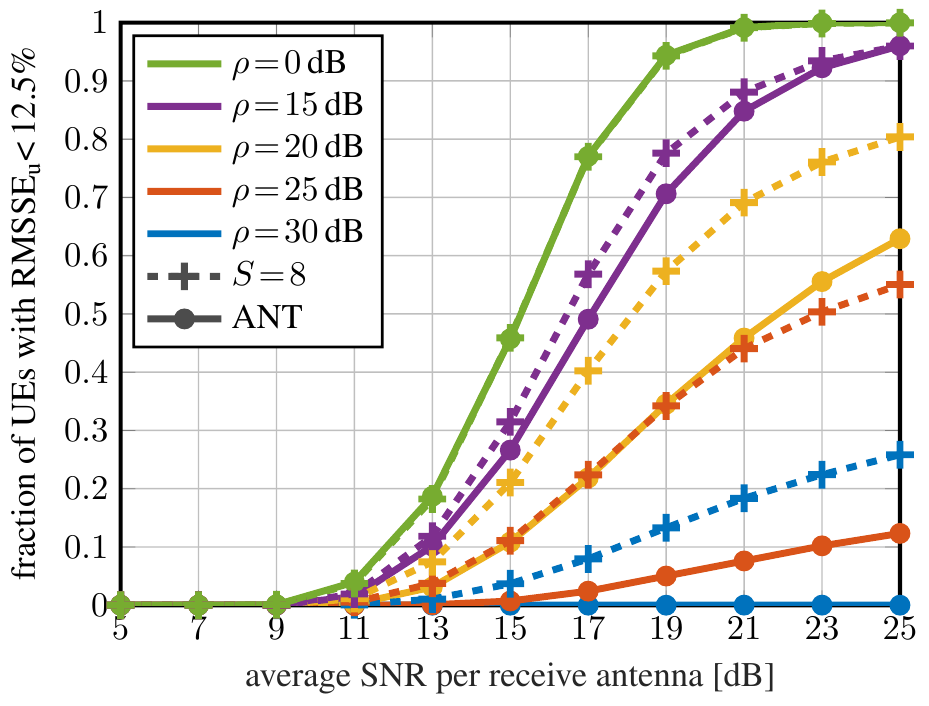}
\label{fig:sweep_jammer:lmmse_nlos}
}
\caption{Fraction of successfully served UEs vs.~SNR, compared between antenna domain (ANT; solid curves) and beam-slicing with $S=8$ (dashed curves) for different relative jammer powers $\rho$ [dB]. The ADC resolution is $q=4$ bits per real dimension.}
\label{fig:sweep_jammer}
\vspace{-0.15cm}
\end{figure*}

\subsection{The Efficacy of Beam-Slicing} \label{sec:res_eff}
In Figures~\ref{fig:sweep_S:ber} and \ref{fig:sweep_S:cdf}, we evaluate the performance of SNIPS (solid) and CHOPS (dashed) for different antenna cluster sizes~$S$. 
We compare the baselines, which perform soft-nulling or orthogonal projection in antenna domain (ANT), against SNIPS and CHOPS with cluster sizes $S\in\{2,4,8,16,32,64, 256\}$, hence considering analog beam-slicing that only operates on a pair of adjacent antennas up to a single cluster consisting of the whole antenna array.
We note that with a cluster size $S=B=256$, beam-slicing corresponds to performing a full beamspace transform.
For these experiments, we consider a strong relative jammer power $\rho=25$\,dB, and a BS with  $q=4$ bit ADCs (per real dimension).

\fref{fig:sweep_S:ber:los} shows uncoded BER results under LoS conditions.
We see that beam-slicing with two-antenna clusters ($S=2$) already yields noticeable BER improvements over the antenna-domain baselines. 
SNIPS and CHOPS are virtually identical for small clusters, but SNIPS slightly outperforms CHOPS for large clusters.
\fref{fig:sweep_S:ber:los} also shows that large clusters outperform small ones. 
However, the performance of a full beamspace transform ($S=256$) is inferior to $S=64$, which exhibits the best performance for both  SNIPS and CHOPS. The reason for this has nothing to do with how the beamspace transform distributes the jammer interference to the ADCs. Instead, the performance decrease can be addressed to the fact that, after a full beamspace transform, the UE signals are concentrated to only a few ADCs, so that low-resolution ($q \leq 4$) ADCs can no longer represent them accurately.
This observation suggests that the fully-centralized beamspace transform, which in any case is impractical, may not necessarily be optimal for achieving the full potential of beam-slicing. 

\fref{fig:sweep_S:ber:nlos} shows results for non-LoS conditions, where the channels are less sparse (in beamspace domain) than under LoS conditions. As a consequence, the gains obtained by beam-slicing over the antenna-domain baselines are not as pronounced as for LoS conditions.
Nonetheless, beam-slicing with $S=8$ still offers a gain of 4\,dB at a BER of 1\% compared to the baselines. 
Due to the decrease in channel sparsity, there is now a strict improvement of performance for larger cluster sizes: Even for a full beamspace transform, the UE signals are sufficiently distributed over the different low-resolution ADCs as to be appropriately represented.
 
The behavior of the fraction of UEs whose $\rmmse_u$ is below $12.5\%$ at a given SNR is shown in \fref{fig:sweep_S:cdf}, where \fref{fig:sweep_S:cdf:los} shows results under LoS conditions and \fref{fig:sweep_S:cdf:nlos} shows results under non-LoS conditions.
In terms of this criterion, the antenna-domain baselines are not able to successfully serve a significant percentage of UEs under LoS conditions, and less than $20\%$ under non-LoS conditions, regardless of SNR.
In contrast, beam-slicing with four antennas per cluster ($S=4$) at high SNR can serve more than $40\%$ of UEs both under LoS and non-LoS conditions. 
This fraction of successfully served UEs increases with the cluster size $S$ (though again not strictly monotonically under LoS conditions), with SNIPS and CHOPS being able to serve more than $80\%$ of UEs at high SNR for $S\in\{32,64\}$ under LoS conditions, and more than $65\%$ of UEs under \mbox{non-LoS} conditions. 

The general picture that emerges from Figures~\ref{fig:sweep_S:ber} and \ref{fig:sweep_S:cdf} is that beam-slicing significantly improves jammer mitigation already for modestly sized (and hence practical) antenna clusters. 
The strongly similar performance of SNIPS and CHOPS suggests that the advantage afforded by beam-slicing is quite independent from the digital jammer mitigation method used, so that beam-slicing may successfully be combined with and enhance a variety of digital jammer mitigation methods in systems that rely on low-resolution ADCs. 

\subsection{When is Beam-Slicing Needed?}
The experiments in \fref{sec:res_eff} indicate that, for a strong jammer, SNIPS and CHOPS outperform their antenna-domain baselines and that the best performance is achieved with a large cluster size $S$.
However, as discussed in \fref{sec:prior_work}, analog spatial transforms spanning a large number of antennas are difficult to implement in practice~\cite{JoshiThesis,vallappil21a,guo21a,hall90a,tudosie09a}, with a full beamspace transform being probably infeasible for massive MU-MIMO systems. We therefore consider SNIPS and CHOPS with a moderately-sized antenna cluster of $S=8$ for our subsequent evaluations.
 
In \fref{fig:sweep_ADC}, we analyze the impact of ADC resolution when the relative jammer power is $\rho=25$\,dB for CHOPS under LoS conditions (\fref{fig:sweep_ADC:pos_los}), as well as for SNIPS under LoS (\fref{fig:sweep_ADC:lmmse_los}) and non-LoS (\fref{fig:sweep_ADC:lmmse_nlos}) conditions.\footnote{Because of the virtually identical performance of CHOPS and SNIPS under LoS conditions shown by Figures~\ref{fig:sweep_ADC:pos_los} and \ref{fig:sweep_ADC:lmmse_los}, we have omitted a plot for CHOPS under non-LoS conditions. The same applies for \fref{fig:sweep_jammer}.}
We consider the fraction of UEs successfully served in terms of the criterion $\rmmse_u<\!12.5$\%.
For infinite- or high-resolution (\mbox{$q=8$}) ADCs, the beam-slicing methods have identical performance as their antenna-domain baselines in all the setups. However, already for 6-bit ADCs, SNIPS and CHOPS outperform their antenna-domain counterparts. For low-resolution ADCs with $q\leq4$, the antenna-domain methods are unable to serve a significant fraction of UEs (less than~$2$\%) under LoS conditions. In contrast, SNIPS and CHOPS can at least serve some UEs at high SNR even for $q=3$, and they can serve up to $65$\% of UEs for $q=4$ under LoS conditions. Under non-LoS conditions, the beam-slicing gains are less pronounced but still significant.

In \fref{fig:sweep_jammer}, we consider the impact of the relative jammer power $\rho$ for 4-bit ADCs. We again compare CHOPS (LoS conditions; \fref{fig:sweep_jammer:pos_los}) and SNIPS (LoS and non-LoS conditions; Figures~\ref{fig:sweep_jammer:lmmse_los} and \ref{fig:sweep_jammer:lmmse_nlos}) against their antenna-domain counterparts. We see that beam-slicing does not offer a performance gain with weak jammers that are only as strong as the average UE ($\rho=0$\,dB). However, for a jammer with $\rho=15$\,dB, SNIPS and CHOPS already significantly outperform their antenna-domain baselines in terms of successfully served UEs, and this gap continues to widen as jamming power increases further. 

Together, these experiments confirm that strong jammers pose a serious problem for classical all-digital jamming suppression methods when combined with low-resolution ADCs. Our results also show that beam-slicing can successfully mitigate this problem in a practical manner. 

\begin{figure*}[t!]
\centering
\subfigure[BER]{
\includegraphics[width=0.825\columnwidth]{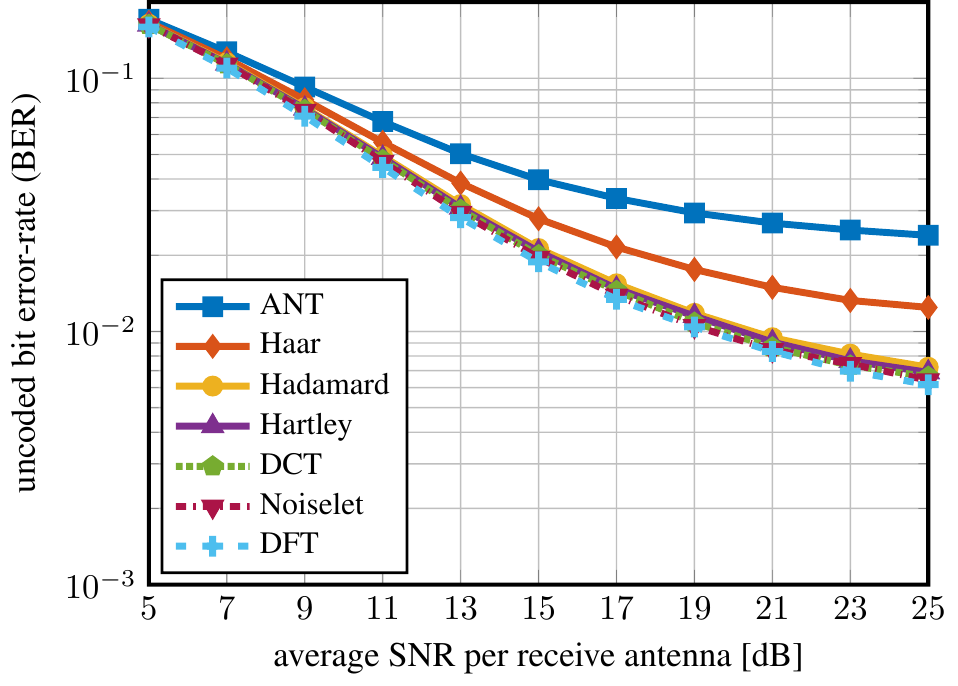}
\label{fig:different_transforms:ber}
}
\!\!\quad\quad\quad\quad\quad 
\subfigure[Fraction of served UEs]{
\includegraphics[width=0.8\columnwidth]{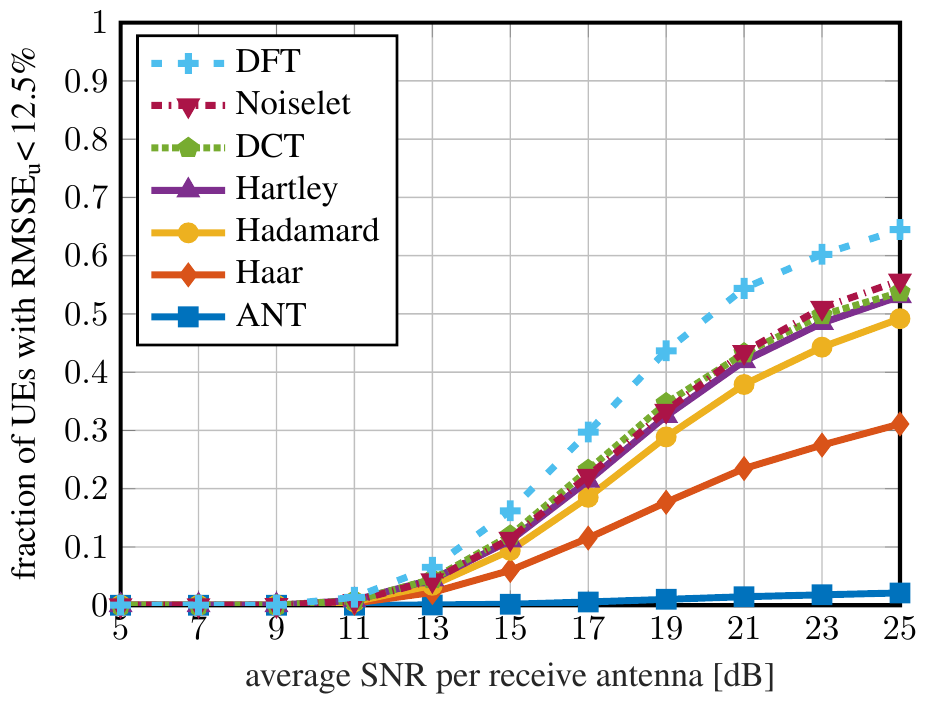}
\label{fig:different_transforms:cdf}
}
\caption{Performance comparison for SNIPS with different ``beam-slicing'' transforms $\bT$. The size of the transform clusters is $S=8$, the relative jammer power is $\rho=25$\,dB, the ADC resolution is $q=4$ bits per real dimension, and the channels are LoS.}
\label{fig:different_transforms}

\end{figure*}

\subsection{Ablation Studies}
\label{sec:ablation}

After showing the general efficacy of beam-slicing and analyzing the conditions under which it leads to performance improvements, we now justify some of the choices in our beam-slicing design empirically. Specifically, we compare the choice of the DFT against other spatial transforms; we show the utility of the increased angular diversity provided by rotating the DFTs as in \eqref{eq:Vc}; and we provide evidence that using steadily-progressing rotation angles as in \eqref{eq:Vc} is close-to-optimal.
For this, we generalize the per-cluster beam-slicing transform \eqref{eq:Vc}~to
\begin{align}
	\bV_c(\bT, \phi_c) = 
	\bT~\text{diag}\!\left(1, 
		\dots, e^{-i\phi_c(s-1)},
	\dots, e^{-i\phi_c(S-1)}
	\right),  \label{eq:general_Vc}
\end{align}
resulting in the overall beam-slicing matrix
\begin{align}
	\bV(\bT,\boldsymbol{\phi}) &= \text{diag}\big(\bV_1(\bT, \phi_1), \dots, \bV_C(\bT, \phi_c)\big),
\end{align}
where $\boldsymbol{\phi}=[\phi_1,\dots,\phi_C]$ are the per-cluster rotation angles, and $\bT\in\mathbb{C}^{S\times S}$ is an arbitrary unitary transform. We restrict ourselves to unitary transforms $\bT$ since this ensures the unitarity of $\bV(\bT,\boldsymbol{\phi})$, which in turn ensures that the beam-sliced noise vector $\hat\bmn=\bV(\bT,\boldsymbol{\phi})\bmn$ has the same distribution as $\bmn$. 

Our previous results have shown the similar-to-identical performance of SNIPS and CHOPS over a wide range of parameters. For simplicity, we therefore restrict our analysis from here on to SNIPS. We consider SNIPS with 4-bit ADCs,  under a relative jammer power $\rho=25$\,dB in LoS transmission.

In \fref{fig:different_transforms}, we compare the performance for different choices of the transform $\bT$ and uniformly-strided cluster rotations $\boldsymbol{\phi}=(2\pi/B)\times[0,1,\dots,C-1]$ as in \eqref{eq:Vc}. 
The cluster size (i.e., the size of the transforms) is $S=8$.
In addition to the proposed DFT, we consider the Haar transform \cite{kaiser1998fast}, the Hadamard transform \cite{pratt1969hadamard}, the discrete Hartley transform \cite{hartley1942more, bracewell1983discrete}, the discrete cosine transform (DCT) \cite{ahmed1974discrete}, and the Noiselet transform \cite{coifman2001noiselets} as candidates for $\bT$.
As a baseline, we also include the performance without beam-slicing, i.e., when operating directly in the antenna domain.  
We see that even the worst-performing Haar transform significantly outperforms the antenna-domain baseline on both performance metrics.
In terms of the uncoded BER shown by \fref{fig:different_transforms:ber}, all the other considered transforms $\bT$ yield similar performance, with a slight advantage for the DFT.
Nevertheless, when looking at the fraction of served UEs in \fref{fig:different_transforms:cdf}, we are able to appreciate a larger performance gap between the DFT and the other transforms.
Together, these results support the choice of the DFT in our beam-slicing design when merely considering system performance.

Despite the foregoing analysis, the choice of the beam-slicing transform may be guided by different factors in practice.
For example, while large DFTs are hard to implement using analog circuitry (e.g., $S=16$ is the largest Butler matrix DFT reported in the open literature; cf. \fref{sec:prior_work}), high-dimensional Hadamard transforms can be implemented with particular ease---see, e.g., the recent work in~\cite{Enciso20} with an implementation for $S=128$.
Therefore, one could imagine preferring a subpar transform, such as the discrete Hadamard transform, but with a large cluster size $S$, over the DFT with a lower cluster size when implementing beam-slicing.
A detailed study of such trade-offs is, however, left for future work. 
\begin{figure}[t!]
\centering
\includegraphics[width=0.8\columnwidth]{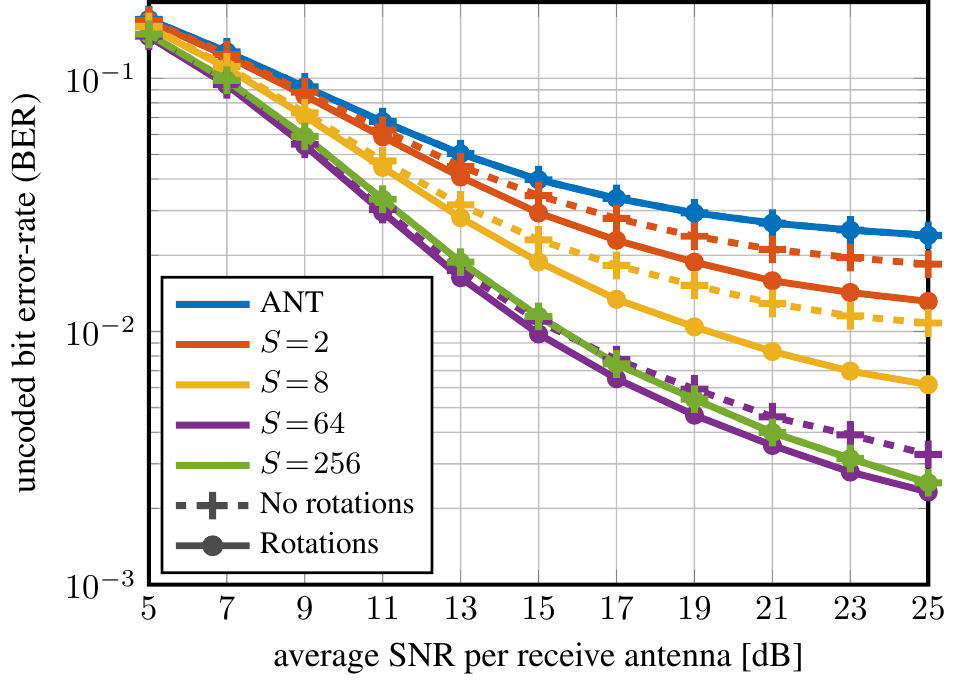}
\caption{Performance comparison between SNIPS with uniformly-strided cluster rotations and SNIPS without cluster rotations for different cluster sizes $S$. The relative jammer power is $\rho=25$\,dB, the ADC resolution is $q=4$ bits per real dimension, and the channels are LoS.}
\label{fig:rotations_vs_norotations}
\vspace{-0.15cm}
\end{figure}
\begin{figure}[t!]
\subfigure[Training set (10$^3$ channels)]{
\hspace{-3mm}
\includegraphics[width=0.5\columnwidth]{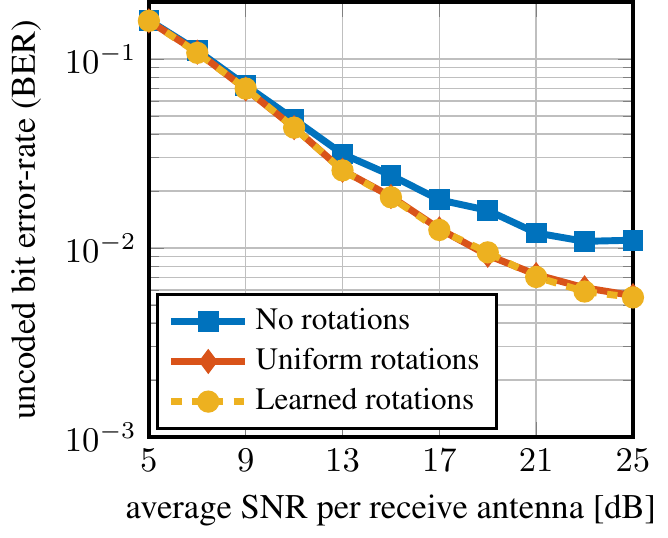}
\label{fig:learned_rotations:train}
}
\hspace{-3mm}
\subfigure[Test set (10$^4$ channels)]{
\includegraphics[width=0.5\columnwidth]{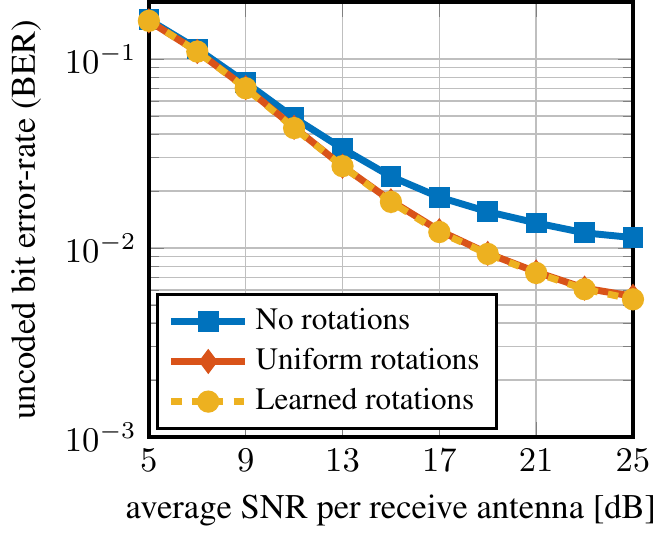}
\label{fig:learned_rotations:test}
\hspace{-3mm}
}
\caption{Performance comparison between SNIPS with uniformly-strided cluster rotations, SNIPS with learned cluster rotations, and SNIPS without cluster rotations. The size of the transform clusters is $S=8$, the relative jammer power is $\rho=25$\,dB, the ADC resolution is $q=4$ bits per real dimension, and the channels are LoS.}
\label{fig:learned_rotations}
\end{figure}

In \fref{fig:rotations_vs_norotations}, we fix $\bT$ to be the DFT, and analyze the performance gain of uniformly-strided cluster rotations $\boldsymbol{\phi}=(2\pi/B)\times[0,1,\dots,C-1]$ as in \eqref{eq:Vc}, over fixed (unrotated) clusters, which corresponds to $\boldsymbol{\phi}=\mathbf{0}$ in \eqref{eq:general_Vc}.
In our analysis, we consider different cluster sizes $S$.
Note that for the antenna domain ($S=1$) and beamspace transform ($S=256$),~both cases (with and without rotations) are mathematically equivalent, so the BER curves coincide.
For all other choices, the inclusion of rotations provides significant performance gains.

For a better understanding of how uniformly-strided cluster rotations $\boldsymbol{\phi}=(2\pi/B)\times[0,1,\dots,C-1]$ increase angular diversity, we make the following remarks:
Each row of the $S$-point DFT $\bF_S$ corresponds to a sampled complex sinusoid which, up to a real-valued scalar factor, can be represented as $f_\omega[s]=e^{-i\omega s}$, $s=0,\ldots,S-1$, where the angular frequency $\omega$ is fixed per row.
Across the rows of $\bF_S$, the angular frequency $\omega$ increases in steps of $2\pi/S$.
So the rows of the full beamspace transform $\bF_B$ are sampled complex sinusoids whose angular frequencies increase in steps of $2\pi/B$.
In \fref{eq:general_Vc}, we multiply $\bF_S$ with a diagonal matrix with entries \mbox{$f_\phi[s]=e^{-i\phi s}$, $s=0,\ldots,S-1$.}
For each row of $\bF_S$, this corresponds to point-wise multiplication of $f_\omega[s]$ and $f_\phi[s]$, resulting in an effective angular frequency of $\omega+\phi$.
By choosing the diagonal entries corresponding to the $c$th cluster as $f_\phi[s]$ with $\phi=(2\pi/B)(c-1)$, we obtain, across all clusters, signals whose angular frequencies increase in steps of $2\pi/B$, as they would for the full beamspace transform $\bF_B$.
We thereby achieve the same angular diversity as the beamspace transform (though not the same sharpness in spatial resolution).

While the choice of uniformly-strided cluster rotations \mbox{$\boldsymbol{\phi}=(2\pi/B)\times[0,1,\dots,C-1]$} comes naturally for the DFT, this does not yet guarantee their optimality.
To gain some insight into the question of optimality, we used a data-driven approach for learning the per-cluster rotations $\boldsymbol{\phi}$.
Specifically, we used a coordinate descent algorithm to determine better rotation angles that proceeds as follows: For cluster $c$, we fix all per-cluster rotations in $\boldsymbol{\phi}$ except for $\phi_c$. 
The angle $\phi_c$ is then swept over a grid of $148$ possible rotation angles between $0$ and $(2\pi/B)C$, evaluating for every grid point the uncoded BER of SNIPS (with these rotation angles $\boldsymbol{\phi}$) at an SNR of $20$\,dB and a relative jammer power $\rho=25$\,dB over a training set consisting of $10^3$ LoS channels (one transmit symbol per UE per channel).
We then fix $\phi_c$ to the rotation angle with the lowest uncoded BER, and the procedure is repeated analogously for the next cluster rotation angle $\phi_{c+1}$.
This coordinate descent process is repeated (for each rotation angle $\phi_c,\,c=1,\dots,C$ in $\boldsymbol{\phi}$)  for a total of $50$ iterations.

\fref{fig:learned_rotations} shows uncoded BER simulation results for the rotation angles learned with this coordinate-descent learning algorithm.
In \fref{fig:learned_rotations:train}, we compare the BER performance of the learned rotation angles with the performance of uniformly strided rotations $\boldsymbol{\phi}=(2\pi/B)\times[0,1,\dots,C-1]$ (and no rotations), considering only the channels contained in the training set.
We see that the learned rotations offer virtually no improvement over uniform rotations.
Since learned rotations do not outperform uniform rotations even on the training set, it would seem unlikely that they offer improvements when evaluated on new channels. This is confirmed in \fref{fig:learned_rotations:test}, where we compare the performance for $10^4$ new LoS channels.

\section{Conclusions}

We have proposed a novel method to mitigate strong jamming attacks in mmWave massive MU-MIMO BSs relying on low-resolution ADCs.
Concretely, we have shown that strong jammers force the quantization range of low-resolution ADCs to drown the UE signals in quantization noise, thereby introducing distortions which are difficult to remove with digital signal processing.
In order to enable effective and practical jammer mitigation in systems that rely on low-resolution ADCs, we have proposed \emph{beam-slicing}, a non-adaptive, distributed analog spatial transform. Beam-slicing exploits the strong directionality of mmWave signals to focus the jammer's energy on a subset of ADCs. An estimate of each ADC-output's fidelity can then be exploited to estimate the UE signals on the basis of the interference-free ADC outputs. 
We have proposed two such estimation methods, SNIPS and CHOPS, which differ in how they cancel the jamming signal. SNIPS performs LMMSE equalization on the beam-sliced signal to soft-null the jammer; CHOPS projects the beam-sliced signal on the orthogonal subspace to remove the jammer completely. Both of these methods leverage Bussgang's decomposition to model quantization artifacts.
We have shown using simulations that beam-slicing significantly improves jammer mitigation already for small transform clusters, with SNIPS and CHOPS clearly outperforming their antenna-domain counterparts. 
This suggests that beam-slicing is a practical method that may be combined with and enhance a variety of digital jammer mitigation methods to provide increased robustness of low-resolution mmWave massive MU-MIMO BSs to jamming attacks.

\appendices
\section{Estimating the Projection matrix} \label{app:projection}

We offer a brief exposition on why \eqref{eq:pest} is a sensible estimate
for the projection matrix
$\bP = \bI_B - \hat\Hj \herm{\hat\Hj} / \|\hat\Hj\|^2$. 
In doing so, we ignore the thermal noise and the quantization noise. 
Under these simplifications, the jammer sequence receive matrix \eqref{eq:jammer_quantiz} can be rewritten as
\begin{align}
	\bR_J = \hat\bmj [\sj_1, \dots, \sj_N] = \hat\bmj \herm{\mathbf{\sj}}. \label{eq:pure_jammer_sequence}
\end{align}

\begin{lem}
If the jammer receive sequence is given as \eqref{eq:pure_jammer_sequence} with $\bR_J\!\neq\mathbf{0}$, and if $\EIf= \bR_J\herm{\bR_J}$, then the orthogonal projection~on-to the orthogonal complement of the subspace spanned by~$\hat\Hj$~is
\begin{align}
	P: \mathbb{C}^B \to \mathbb{C}^B, \bmx \mapsto \textnormal{\Pest}\bmx
\end{align}
with
\begin{align}
	\textnormal{\Pest} = \bI_B - \frac{\EIf}{\tr{(\EIf)}}.
\end{align}
\end{lem}
\begin{proof}
We have 
\begin{align}
\EIf = \bR_J\herm{\bR_J} = 	\hat\bmj \herm{\mathbf{\sj}} \mathbf{\sj} \herm{\hat\bmj} = \|\mathbf{\sj}\|^2 \hat\bmj \herm{\hat\bmj}
\end{align}
and $\tr{(\EIf)} =  \|\mathbf{\sj}\|^2 \|\hat\bmj\|^2$. From this, the claim follows by recalling \cite[Sec.\,2.6.1]{GV96}.
\end{proof}

\balance

\end{document}